\documentclass[11pt]{article}

\usepackage{fullpage}
\usepackage{amsmath}
\usepackage{amsthm}
\usepackage{amssymb}
\usepackage{amsfonts}
\usepackage{arcs}
\usepackage{float}
\usepackage{graphicx}
\usepackage[small]{caption}
\usepackage{color}
\usepackage{enumitem}
\usepackage{marvosym}
\usepackage{cases}
\usepackage{url}
\usepackage{hyperref}
\usepackage{pifont}
\usepackage{lineno}
\usepackage{graphicx}
\usepackage{lscape}

\usepackage{tikz}
\usepackage{url}
\usepackage{float}
\usepackage{multirow}
\usepackage{pgfplots} 
\usepackage{pdflscape}
\usepackage{algorithm,algorithmic}
\usepackage{cite}
\usepackage{hyperref}

\usepackage{microtype}

\usepackage{tcolorbox}
\usepackage{listings}

\newcommand{\change}[1]{\textcolor{black}{#1}}

\newcommand{\etal}{{et~al.}}

\pgfplotsset{compat=1.15} 

\newtheorem{theorem}{Theorem}

\begin{document}
	
	\title{\textsc{Experiments with Unit Disk Cover Algorithms for Covering Massive Pointsets}\thanks{A preliminary version of this paper appeared in the Proceedings of the  International Symposium on Experimental Algorithms, Springer, Cham, 2019~\cite{ghosh2019unit}. Research on this paper was fully supported by the University of North Florida Academic Technology Grant, and partially by the NSF Grant CCF-1947887.}}

	\author{Rachel Friederich \and Anirban Ghosh    \and
		Matthew Graham \and Brian Hicks \and Ronald Shevchenko 
	}
	
	\author{
		Rachel Friederich \\
		School of Computing \\
		University of North Florida\\
		Email: \texttt{n01140328@unf.edu}           
		\and
		Anirban Ghosh \\
			School of Computing \\
			University of North Florida\\
			Email: \texttt {anirban.ghosh@unf.edu}      
			\and
		Matthew Graham \\
		School of Computing \\
		University of North Florida\\
			Email: \texttt{n00612546@unf.edu} 
			\and
		Brian Hicks \\
		School of Computing \\
		University of North Florida\\
			Email: \texttt{n00133251@unf.edu}     
			\and
		Ronald Shevchenko \\
		School of Computing \\
		University of North Florida\\
			Email: \texttt{n01385011@unf.edu}  
	}

\maketitle

\begin{abstract}
	
Given a set of $n$ points in the plane, the Unit Disk Cover (UDC) problem asks to compute the minimum number of unit disks required to cover the points, along with a placement of the disks. The problem is NP-hard and several approximation algorithms have been designed over the last three decades. In this paper, we have engineered and experimentally compared practical performances of some of these algorithms on massive pointsets. The goal is to investigate which algorithms run fast and give good approximation in practice. 

We  present a simple $7$-approximation algorithm for UDC that runs in $O(n)$ expected time and uses $O(s)$ extra space, where $s$ denotes the size of the generated cover. In our experiments, it turned out to be the speediest of all. We also present two heuristics to reduce the sizes of  covers generated by it without slowing it down by much.

To our knowledge, this is the first work that experimentally compares geometric covering algorithms. Experiments with them using massive pointsets (in the order of millions) throw light on their practical uses. We share the engineered algorithms via \textsf{GitHub}\footnote{\url{https://github.com/ghoshanirban/UnitDiskCoverAlgorithms}} for broader uses and future research in the domain of geometric optimization.  
\end{abstract}


\section{Introduction} \label{sec:intro}

Geometric covering is a well-researched family of fascinating optimization problems in computational geometry and has been studied for decades. To date, research has been confined mostly to the theoretical arena only. Among these problems, the {Unit Disk Cover} (UDC) problem has turned out to be one of the fundamental covering problems. Given a set $P$ of $n$ points $p_1,\ldots,p_n$ in the Euclidean plane, the UDC problem asks to compute the minimum number of possibly intersecting unit disks (closed disks of unit radius)  required to cover the points in $P$, along with a placement of the disks. See Fig.~\ref{fig:f2} for an example. Since the algorithms for UDC can be easily scaled for covering points using disks of any fixed radius $r>0$, for the sake of brevity, we use $r=1$.

 The UDC problem has interesting applications in wireless networking, facility location, robotics, image processing, and machine learning. For instance, $P$ can be perceived as a set of clients or locations of interest seeking service from service providers, which can be modeled using a set of fixed-radius disks. The goal is to provide service or cover these locations using the minimum number of service providers. 
 
 \begin{figure}[h]
    \centering
    \includegraphics[scale=0.8]{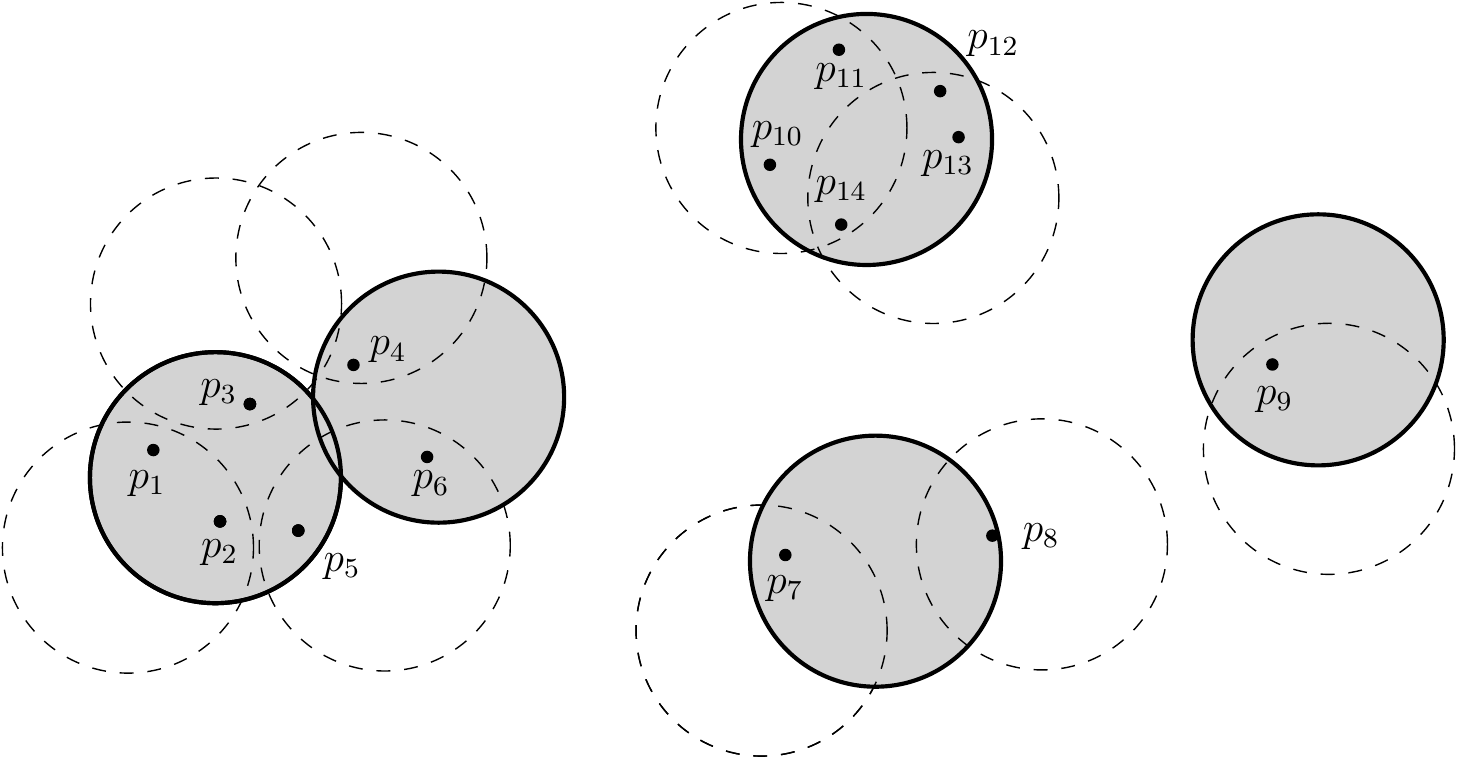}
    \caption{Any optimal solution for the pointset $P=\{p_1,\ldots,p_{14}\}$ contains exactly $5$ disks; an optimal solution for $P$ is shown using gray disks. A non-optimal solution is shown using a set of $9$ dashed disks.}
    \label{fig:f2}
\end{figure}

The UDC problem has a long history. Back in 1981, UDC was shown to be NP-hard by Fowler~\cite{fowler1981optimal}. The first known approximation algorithm for UDC is a PTAS designed by Hochbaum and Maass \cite{hochbaum1985approximation} that runs in $O(\ell^4(2n)^{4\ell^2+1})$ time having an approximation factor of $(1 + \frac{1}{\ell})^2$, for any integer $\ell\geq 1$. Gonzalez~\cite{gonzalez1991covering} presented two approximation algorithms; a $2(1+\frac{1}{\ell})$-approximation algorithm that runs in $O(\ell^2n^7)$ time, where $\ell$ is a positive integer and another $8$-approximation algorithm with a runtime of $O(n\log |\text{OPT}|)$, where $|\text{OPT}|$ is the number of disks in an optimal cover. Charikar, Chekuri, Feder, and Motwani~\cite{charikar2004incremental} devised a $7$-approximation algorithm for the UDC problem (the authors used the name \textsc{Dual Clustering} for this problem). A $O(1)$-approximation algorithm with a runtime of $O(n^3\log n)$ is presented by Br{\"o}nnimann and Goodrich \cite{bronnimann1995almost}. 
Franceschetti, Cook, and Bruck~\cite{franceschetti2001geometric}~developed an algorithm with an approximation factor of $3(1+\frac{1}{\ell})^2$ having a runtime of $O(Kn)$, where $\ell$ is a positive integer and $K$ is a constant that depends on $\ell$. A $2.8334$-approximation algorithm is designed by Fu, Chen, and Abdelguerfi~\cite{fu2007almost} that runs in $O(n(\log n \log \log n)^2)$ time. 
Liu and Lu~\cite{liu2014fast} designed a $25/4$-approximation algorithm having a runtime of  $O(n\log n)$. Biniaz, Liu, Maheshwari, and Smid~\cite{biniaz2017approximation}~devised a $4$-approximation algorithm that has a runtime of $O(n\log n)$. Recently, Dumitrescu, Ghosh, and T{\'o}th~\cite{dumitrescu2020online} have designed an online $5$-approximation\footnote{In the literature of online algorithms, the term \emph{competitive ratio} is used instead of approximation factor.} algorithm for the problem. 

In the era of Big Data, the sizes of spatial data sets are growing exponentially. Thus, finding good quality solutions efficiently for NP-hard geometric optimization problems has posed a new challenge for algorithm engineers. In this regard, because of the practical importance of the UDC problem, we believe that it is worthwhile to investigate which algorithms designed for UDC are the best for processing massive pointsets in practice. 

Covering problems involving points and disks are well-studied in computational geometry; see for instance,~\cite{aloupis2012covering,  agarwal2014near, bar2013note, chazelle1986circle,das2012discrete, de2009covering, dumitrescu2018computational, kaplan2011optimal, liao2010polynomial, guo2021geometric}. Bus, Mustafa, and Ray designed a practical algorithm for the geometric hitting set problem; see~\cite{bus2018practical}. The UDC problem has also been considered in the streaming setup by Liaw, Liu, and Reiss~\cite{liaw2017}.

\paragraph{Our contributions.} For our experiments, we have implemented the following algorithms; appropriate abbreviations using the authors' names and dates of publication are used for naming purposes.
\begin{enumerate}   \itemsep0pt
 
    \item G-1991 by Gonzalez (1991)~\cite{gonzalez1991covering}
    \item CCFM-1997 by Charikar, Chekuri, Feder, and Motwani (1997)~\cite{charikar2004incremental}
   
    \item LL-2014 by Liu and Lu (2014)~\cite{liu2014fast}
    \item BLMS-2017 by Biniaz, Liu, Maheshwari, and Smid (2017)~\cite{biniaz2017approximation}
    \item DGT-2018 by Dumitrescu, Ghosh, and T\'{o}th (2018)~\cite{dumitrescu2020online}
\end{enumerate}
We have refrained from implementing the algorithms from~\cite{bronnimann1995almost,fu2007almost, hochbaum1985approximation, franceschetti2001geometric} since they are not practical and mainly of theoretical interest.

We present a simple $7$-approximation algorithm named \textsc{FastCover} that runs in $O(n)$ expected time; see Section~\ref{GHS}. In our experiments, we found \textsc{FastCover} to be the fastest of all. We also present two heuristics that effectively help  to reduce the sizes of  covers generated by it. \textsc{FastCover} with the first heuristic included is named  \textsc{FastCover}\texttt{+} and the one in which both the heuristics are included is named  \textsc{FastCover}\texttt{++}. To our surprise, we found that in some cases \textsc{FastCover}\texttt{++}  could beat some of the sophisticated algorithms in speed and solution quality simultaneously. These three versions of \textsc{FastCover} behave like three optimization levels for the algorithm, where  \textsc{FastCover} being the fastest and  \textsc{FastCover}\texttt{++} the slowest in practice. Wherever possible, \textsc{FastCover}\texttt{++} produces the smallest covers among these three. 

In our experiments, we have used both synthetic and real-world massive pointsets. The largest pointset used in the experiments contains $\approx10.8$ million points. 
The algorithms are implemented in C\texttt{++}17 using the CGAL 5.3 library~\cite{cgal:eb-21b}. 
For broader uses of these algorithms, we share our code via \textsf{GitHub}. 

In our knowledge, this is the first work that experimentally compares the existing algorithms for UDC. Experiments with them using massive pointsets throw light on their practical uses. 

In Section~\ref{sec:algorithms}, we discuss the algorithms implemented in this paper along with the \textsc{FastCover} algorithm. In Sections~\ref{sec:exp}, we present our experimental results including tables and plots.  In Section~\ref{sec:con}, we present our recommendations and conclusions.  

\paragraph{Notations and terminology.} We denote a point $p \in \mathbf{R}^2$ using a pair of real numbers $(a,b)$. By $p_x$ and $p_y$, we denote its $x$ and $y$-coordinates, respectively. A \emph{unit ball} is a closed ball of unit radius in $\mathbf{R}^d$. In the plane, we use the term \emph{unit disk}. 

We define the \emph{point density} of  $P$ as the ratio of its size to that of the area of its bounding box.

\section{ Algorithms engineered} \label{sec:algorithms}

In this section, we briefly describe the algorithms we have engineered and provide their pseudocodes along with their asymptotic runtimes. 
To see how the algorithms behave differently, refer to Section~\ref{sec:demo}, where we present the covers generated by the algorithms engineered in this work when run on a $60$-element pointset drawn randomly from a $20 \times 20$ square.

\subsection{G-1991: Gonzalez (1991)}

Gonzalez~\cite{gonzalez1991covering} presented two algorithms for UDC in $d$-space. One of these two algorithms is a PTAS that uses the shifting strategy introduced in~\cite{hochbaum1985approximation}.  This PTAS has an approximation factor of $2(1+\frac{1}{l})^{d-1}$ and runs in $O(\ell^{d-1}d(2\sqrt{d})(\ell\sqrt{d})^{d-1}n^{d(2\sqrt{d})^{d-1}+1})$ time, for every integer $\ell \geq 1$. In the plane, this algorithm has an approximation factor of $2(1+\frac{1}{l})$ and runs in $O(\ell^2 n^{4\sqrt{2}+1})$ time. We did not implement this algorithm due to its high asymptotic runtime. 

\begin{algorithm}[ht]
	\caption{\textsc{: G-1991($P$)}} 
	\label{alg:G-1991} 
	
	\begin{algorithmic}[1] 
		\STATE Let $P_1 := \{p \in P | i_y(p)\text{ is odd} \}$ and $P_2 := \{p \in P | i_y(p)\text{ is even} \}$. Execute the lines 2-12 independently on $P_1$ and $P_2$. The final solution is the union of these two solutions;
		\newline
		\STATE Partition $P_i$ ($i$ is either 1 or 2) w.r.t $i_x(p)$ into sets $S:=S_1,\ldots,S_k$;
		\STATE $R \leftarrow S_1 \cup S_2$; 
		\STATE $j\leftarrow 2$;
		\WHILE{$R\neq \emptyset$}
		\STATE $q \leftarrow \min\{ p_x ~|~p \in R \}$;
		\STATE Let $Q$ be the set of points in $R$ at a distance $\leq \sqrt{2}$ (w.r.t $x$ only) from $q$, $R \leftarrow R \setminus Q$;
		\STATE Include the $\sqrt{2} \times \sqrt{2}$ square in the solution whose left boundary includes $q$ and whose top boundary coincides with the top boundary of the slab having height $D$ that contains $q$;
		\WHILE{$j<k$ \textbf{and} $R$ contains elements from at most one of the sets in $S$}
		\STATE $j \leftarrow j+1$; 
		\STATE $R \leftarrow R \cup S_j$;
		\ENDWHILE
		\ENDWHILE
		\STATE  For every $\sqrt{2} \times \sqrt{2}$ square in the solution, place a unit disk at its center;
	\end{algorithmic}
\end{algorithm}


The other algorithm G-1991, as we call it, has an approximation factor of $2^{d-1}(\lceil \sqrt{d} \rceil)^d$ and runtime of $O(dn + n \log |\text{OPT}|)$, where $|\text{OPT}|$ is the number of disks in an optimal cover. In the plane, G-1991 gives $8$-approximation and runs in $O(n \log |\text{OPT}|)$ time.
See Algorithm~\ref{alg:G-1991} for a high-level description of  G-1991. We use the following notations in the algorithm. Let $p \in P$, then $i_x(p) = \lfloor p_x/\sqrt{2}\rfloor$ and $i_y(p) = \lfloor p_y/\sqrt{2}\rfloor $.

The author presented this algorithm for covering points using axis-parallel squares of fixed size and claimed that the same can be used for UDC. In our implementation, we have used squares of length $\sqrt{2}$ to cover the points and then placed a unit disk at the center of every square. Since a square of length $\sqrt{2}$ can be inscribed inside a unit disk, every point in the input is covered using this approach.

\subsection{CCFM-1997: Charikar, Chekuri, Feder, and  Motwani (1997) }

The algorithm CCFM-1991 by Charikar~\etal~\cite{charikar2004incremental} was originally designed for the online version of UDC. The authors used name \textsc{Dual Clustering} in their paper. In $d$-space, CCFM-1991 gives an approximation factor of $O(2^d d \log d)$.  In $2$-space, CCFM-1997 has an approximation factor of $7$. Refer to Algorithm~\ref{alg:CCFM-1997} for a high-level description of CCFM-1991. No comment was made about its runtime or implementation. 

\begin{algorithm}[H]
	\caption{\textsc{: CCFM-1997($P$)}} 
	\label{alg:CCFM-1997} 
	
	\begin{algorithmic}[1] 
		\STATE Let \texttt{active-centers} and \texttt{inactive-centers} be two empty sets;
		\FOR{$p \in  P$}
		\IF{the distance to the nearest disk center in \texttt{active-centers} $>1$}
		\IF{\texttt{inactive-centers} is empty}
		\STATE Add $p$ to \texttt{active-centers} and add the following six points to \texttt{inactive-centers}: $(p_x+\sqrt{3},p_y), (p_x+\sqrt{3}/2,p_y+1.5), (p_x+\sqrt{3}/2,p_y-1.5), (p_x-\sqrt{3}/2,p_y+1.5), (p_x-\sqrt{3},p_y), (p_x-\sqrt{3}/2,p_y-1.5)$; 
		\STATE \textbf{continue};
		\ENDIF
		\IF{the distance to the nearest disk center $q$ in \texttt{inactive-centers} $\leq 1$}
		\STATE Delete $q$ from \texttt{inactive-centers} and add $q$ to \texttt{active-centers};
		\ELSE
		\STATE Add $p$ to \texttt{active-centers} and add the following six points to \texttt{inactive-centers}: $(p_x+\sqrt{3},p_y), (p_x+\sqrt{3}/2,p_y+1.5), (p_x+\sqrt{3}/2,p_y-1.5), (p_x-\sqrt{3}/2,p_y+1.5), (p_x-\sqrt{3},p_y), (p_x-\sqrt{3}/2,p_y-1.5)$; 
		\ENDIF
		\ENDIF
		\ENDFOR
		\STATE \textbf{return} \texttt{active-centers};
	\end{algorithmic}
\end{algorithm}

\subsection{LL-2014 and LL-2014-1P: Liu and Lu (2014)}

In LL-2014~\cite{liu2014fast}, the plane is divided into vertical strips of width $\sqrt{3}$ each. Inside each strip, we obtain an approximate solution by sorting the points in non-increasing order according to their $y$-coordinates. The next uncovered point inside a strip is covered by placing a disk as low as possible. These disks are placed by centering them on the vertical line that splits the strip into two. The final solution is constructed by taking the union of all the solutions obtained for the strips. This strip system is shifted five times to the right by a distance of $\sqrt{3}/6$ every time. For every shift, we obtain a solution as described above along with one solution for the initial strip configuration. The algorithm returns the best one (having the least number of disks) out of these six solutions. Refer to Algorithm~\ref{alg:LL-2014} for an algorithmic description of LL-2014. The authors show that LL-2014 has an approximation factor of $25/6 \approx 4.17$ and runs in $O(n\log n)$ time.

In our experiments, we also consider a one-pass version of this algorithm since in practice we 
find there is barely any advantage of using six passes instead of one. The authors have used six passes to reduce the approximation factor from $5$ to $25/6$.
We named this one-pass version LL-2014-1P. 
Obviously, this one-pass version is much faster in practice too. 
Since there is no substantial difference between LL-2014 and LL-2014-1P, 
we do not present separate pseudocode 
for this one-pass version.

\begin{algorithm}[h]
\caption{\textsc{: LL-2014($P$)}} 
\label{alg:LL-2014} 

\begin{algorithmic}[1] 
   \STATE $\textsc{Disk-Centers} \leftarrow \emptyset$, \texttt{min} $\leftarrow n+1$;
   
           \STATE Sort $P$ w.r.t $x$-coordinate in $O(n\log n)$ time;
   \FOR{$i \in \{0,1,2,3,4,5\}$}

        \STATE \texttt{current} $\leftarrow 1$, $C \leftarrow \emptyset$, \texttt{right} $\leftarrow P[1]_x + \frac{i\sqrt{3}}{6}$;

        \WHILE{\texttt{current} $\leq n$}
            \STATE \texttt{index} $\leftarrow$ \texttt{current};
            
            \WHILE{$P[\texttt{current}]_x <$ right \textbf{and} \texttt{current} $\leq n$}
                \STATE \texttt{current} $\leftarrow$ \texttt{current} $+$ $1$;
            \ENDWHILE
            
            \STATE \texttt{$x$-of-restriction-line} $\leftarrow$ \texttt{right} $-\sqrt{3}/2$, \texttt{segments} $\leftarrow \emptyset$;
            
            \FOR{$j \leftarrow$ \texttt{index} \textbf{to} \texttt{current}$-1$}
                \STATE $d\leftarrow P[j]_x-$ \texttt{$x$-of-restriction-line}, $y\leftarrow \sqrt{1-d^2}$;
                \STATE Create a segment $s$ having the endpoints $(x\texttt{-of-restriction-line}, P[j]_y+y)$ and $(x\texttt{-of-restriction-line}, P[j]_y-y)$ and insert it into \texttt{segments};
            \ENDFOR
            
            \STATE Sort \texttt{segments} in non-ascending order based on $y$-coordinates of their tops. Greedily stab them by choosing the stabbing point as low as possible, while still stabbing the topmost unstabbed segment. Put the stabbing points (the disk centers) in $C$;
            
            \STATE 
            Increment \texttt{right} by a multiple of $\sqrt{3}$ such that $P[\texttt{current}] - \texttt{right} \leq \sqrt{3}$;

        \ENDWHILE

            \IF{$|C| < $ \texttt{min}} 
                \STATE \textsc{Disk-Centers} $\leftarrow C$, \texttt{min} $\leftarrow |C|$;
            \ENDIF
   \ENDFOR
   \STATE \textbf{return} \textsc{Disk-Centers};
   \end{algorithmic}
\end{algorithm}


\subsection{BLMS-2017: Biniaz, Liu, Maheshwari, and Smid (2017)}

The algorithm BLMS-2017 by Biniaz et al.~\cite{biniaz2017approximation} gives $4$-approximation and runs in $O(n \log n)$ time. Refer to Algorithm~\ref{alg: BLMS-2017} for a high-level description of the algorithm. 

Although the algorithm has a low approximation factor, placing four disks in advance sometimes introduces empty disks in the solutions. For instance, if the distance between any two points in $P$ is greater than $2$, \textsc{BLMS-2017} places exactly $4$ times the optimal number of disks (in this case, $|\text{OPT}|=n$). In comparison, other algorithms such as G-1991 or DGT-2018 place an optimal number of disks. However, in our implementation, we have managed to eliminate such empty disks placed by this algorithm.  BLMS-2017 uses the popular sweep line technique in computational geometry by carefully maintaining a binary search tree of disk centers. For further details, we refer the reader to the original paper~\cite{biniaz2017approximation}. 

\begin{algorithm}[H]
	\caption{\textsc{: BLMS-2017($P$)}} 
	\label{alg: BLMS-2017}
	\begin{algorithmic}[1] 
		\STATE $C \leftarrow \emptyset$, $\textsc{Disk-Centers} \leftarrow \emptyset$;
		\STATE Sort $P$ from left to right in $O(n\log n)$ time;
		\FOR{$p  \in P$ }
		\IF {the nearest point in $C$ is more than $2$ units away from $p$}
		\STATE Place four disks centered at $(p_x,p_y),(p_x+\sqrt{3},p_y),(p_x + \frac{\sqrt{3}}{2},p_y+1.5),(p_x + \frac{\sqrt{3}}{2},p_y-1.5)$ as shown in Fig.~\ref{fig:f1} and 
		add these four points to \textsc{Disk-Centers};
		\STATE $C \leftarrow C \cup \{p\}$;
		\ENDIF
		\ENDFOR
		\STATE \textbf{return} \textsc{Disk-Centers};
	\end{algorithmic}
\end{algorithm}

\begin{figure}[ht]
	\centering
	\includegraphics[scale=0.5]{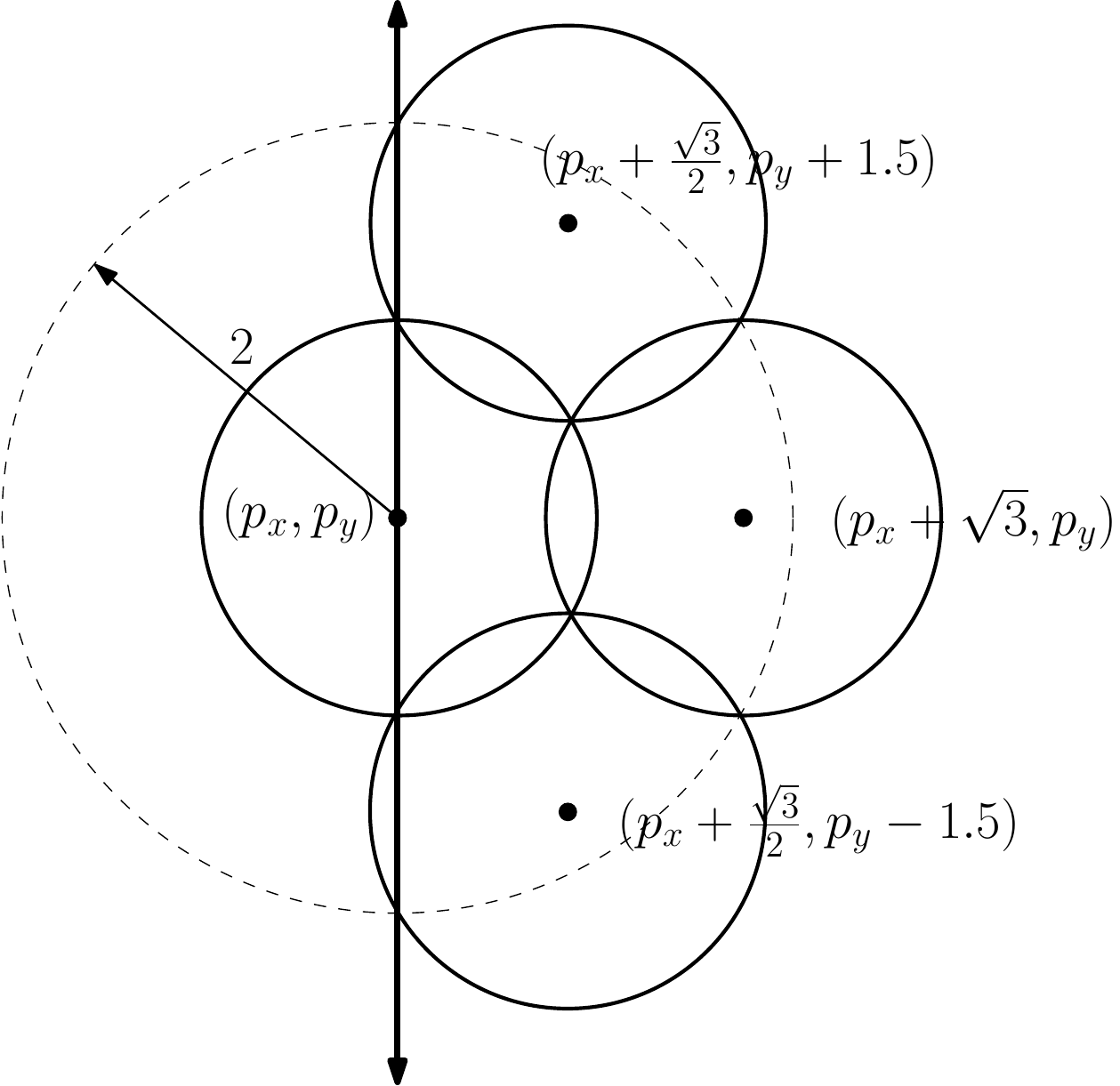}
	\caption{The four disks placed by BLMS-2017 when $p$ is processed. The figure illustrates the situation when the nearest point in $C$ is more than 2 units away from $p$. }
	\label{fig:f1}
\end{figure}

\subsection{DGT-2018: Dumitrescu, Ghosh, and T\'{o}th (2018)}

DGT-2018~\cite{dumitrescu2020online}\footnote{A preliminary version of the paper appeared in the proceedings of the 12th Annual International Conference on Combinatorial Optimization and Applications, 2018 (COCOA 2018), and as such we use the same year in the algorithm abbreviation.} is a simple online algorithm that gives $5$-approximation in the plane. In $d$-space, the algorithm has an approximation factor of $O(1.321^d)$ which is an improvement over CCFM-1997 (an online algorithm).  Refer to Algorithm~\ref{alg: DGT-2018} for a high-level description of this algorithm.

\begin{algorithm}[H]
\caption{\textsc{: DGT-2018($P$)}} 
\label{alg: DGT-2018} 
 
\begin{algorithmic}[1] 
   \STATE $\textsc{Disk-Centers} \leftarrow \emptyset$;
   \FOR{$p \in P$}
        \IF {the distance from $p$ to the nearest point in \textsc{Disk-Centers} is $>1$}
            \STATE \textsc{Disk-Centers} $\leftarrow$ \textsc{Disk-Centers} $\cup$ $\{p\}$; 
        \ENDIF
   \ENDFOR
   \STATE \textbf{return} \textsc{Disk-Centers};
   \end{algorithmic}
\end{algorithm}

\subsection{$\textsc{FastCover}$, \textsc{FastCover}\texttt{+}, and \textsc{FastCover}\texttt{++}} \label{GHS}

 In this section, we present a simple $7$-approximation algorithm \textsc{FastCover} that runs in $O(n)$ expected time. We use a $\sqrt{2}$-sized square grid $\Gamma$. A cell in $\Gamma$ is denoted by $\sigma(i,j)$ where $i,j\in \mathbf{Z}$. Formally, the cell $\sigma(i,j)$ is the intersection of the four half-planes: $x \geq \sqrt{2}i, x < \sqrt{2}(i+1), y \geq \sqrt{2}j, y < \sqrt{2}(j+1)$. For every cell $\sigma(i,j) \in \Gamma$, there exists a unit disk $D(i,j)$ that circumscribes $\sigma(i,j)$.  We say that $D(i,j)$ is the \emph{grid-disk} of $\sigma(i,j)$. Clearly, $D(i,j)$ is unique. 

 The points in the input $P$ are considered sequentially without any kind of pre-processing. Let  $\sigma(i,j)$ be the cell to which the current point $p$ belongs, for some $i,j \in \mathbf{Z}$. Place $D(i,j)$ if not placed previously. 
 
  Since every grid-disk can be represented using a pair of integers, we use a hash-table $\mathcal{H}$ of integer-pairs for storing the placed grid-disks.  The main motivation of using a hash-table is fast lookups and insertions in practice. 
 The cell in which $p$ lies is $\sigma(i,j)$ where $i=\lfloor p_x/\sqrt{2}\rfloor$ and $ j= \lfloor p_y/\sqrt{2}\rfloor$. 
 The center of the grid-disk $D(i,j)$ is located at $(\sqrt{2}i + \frac{1}{\sqrt{2}},\sqrt{2}j + \frac{1}{\sqrt{2}})$. See Algorithm~\ref{alg: fastcover} for a pseudocode. 
 
 \begin{algorithm}[H]
 	\caption{\textsc{: \textsc{FastCover}($P$)}} 
 	\label{alg: fastcover} 
 	
 	\begin{algorithmic}[1] 
 		\STATE \texttt{$\mathcal{H}$} $\leftarrow \emptyset$; 
 		  $\textsc{Disk-Centers} \leftarrow \emptyset$;
 		\FOR{$p \in P$}
 		\STATE  $i\leftarrow \lfloor p_x/\sqrt{2}\rfloor$;   $j\leftarrow \lfloor p_y/\sqrt{2}\rfloor$;
 		
 		\IF{$(i,j) \notin \mathcal{H}$}
 		\STATE insert $(i,j)$ into $\mathcal{H}$ and  $(\sqrt{2}i + \frac{1}{\sqrt{2}}, \sqrt{2}j + \frac{1}{\sqrt{2}})$ into \textsc{Disk-Centers};
 		\ENDIF
 		\ENDFOR
 		\STATE \textbf{return} \textsc{Disk-Centers};
 	\end{algorithmic}
 \end{algorithm}

In the following, we show that \textsc{FastCover} is a $7$-approximation algorithm for UDC.

\begin{theorem}\label{thm:GHS}
	{\textsc{FastCover}} is a $7$-approximation algorithm for the unit disk cover problem that runs in $O(n)$ expected time  using $O(s)$ extra space where $s$ is the size of the solution generated by it. Furthermore, for every integer $n \geq 1$, there exists an $7n$-element pointset for which \textsc{FastCover} places seven times the optimal number of disks. 
\end{theorem}

\begin{proof} The union of all grid-disks in the plane gives $\mathbf{R}^2$. Hence, it is enough to consider the grid-disks for covering the points in $P$. 
	
	Given $i,j \in \mathbf{Z}$, we denote the intersection point of the two lines $x=\sqrt{2}i$ and $y=\sqrt{2}j$ by $p(i,j)$. We refer to these intersection points as \emph{grid point}s. 
	Since the distance between any two grid points is at least $\sqrt{2}$, any unit disk $D$ can contain at most four grid points in $\{p(i,j),p(i+1,j),p(i,j+1),p(i+1,j+1)\}$, for some $i,j \in \mathbf{Z}$. Now observe that $D$ cannot contain exactly three grid points since in that case, $D$ would circumscribe $\sigma(i,j)$ and consequently, $D$ would contain the four grid points $\{p(i,j),p(i+1,j),p(i,j+1),p(i+1,j+1)\}$. 
	
	Consider any disk $D$ from an optimal solution that covers $P$. It suffices to show that to cover the points $D\cap P$, \textsc{FastCover} places at most seven grid-disks. We show this using a proof by cases on the number of grid points contained by $D$. 

	\begin{figure}[ht]
		\centering
		\includegraphics[scale=0.7]{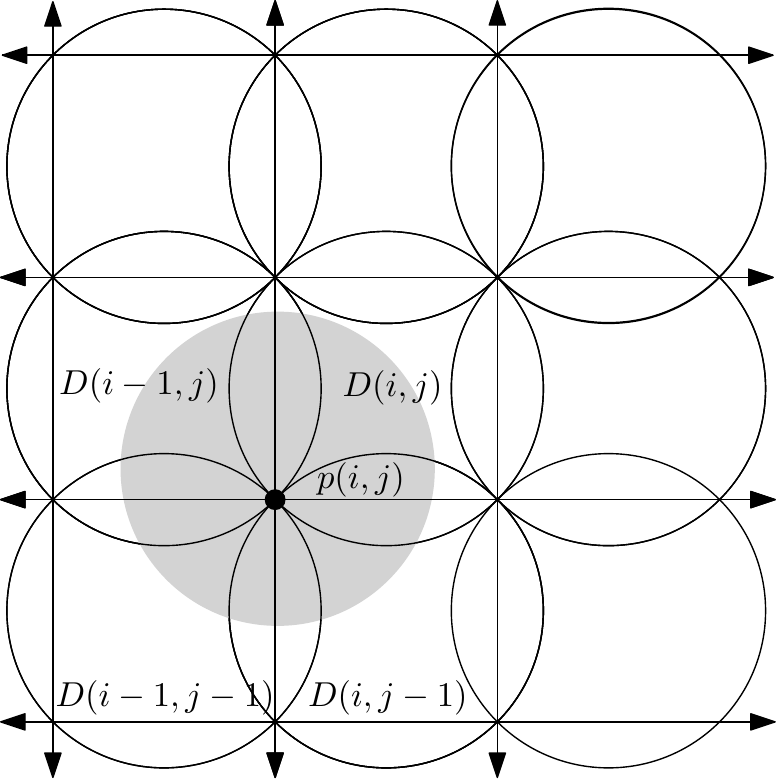} ~~~~~~~~~~~~~~~ \includegraphics[scale=0.7]{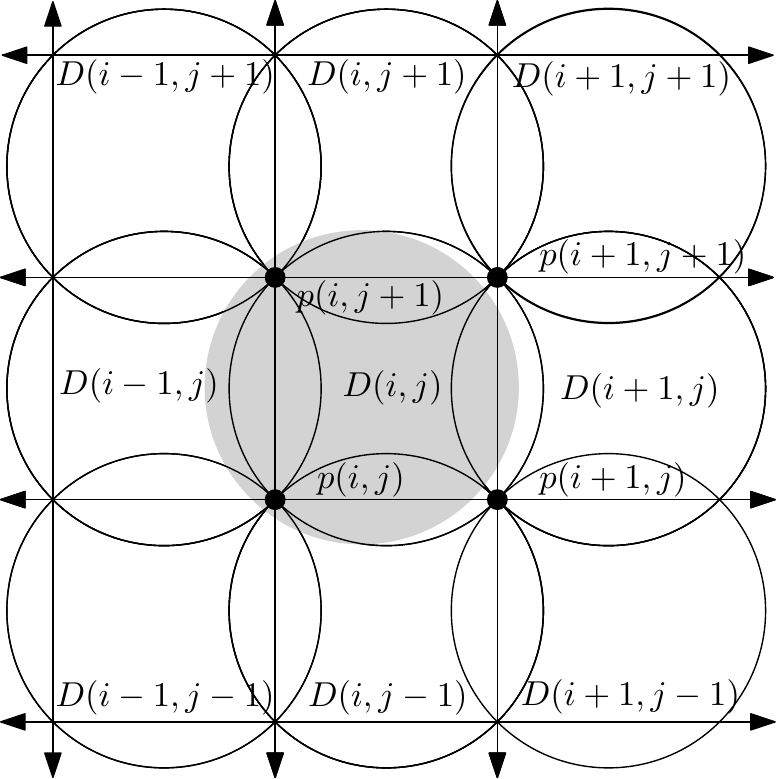} 
		\caption{$D$ is shown in gray. Left: $D$ contains exactly one grid point. Right: $D$ contains exactly two grid points.}
		\label{fig:f7}
	\end{figure}
	
	Assume that $D$ contains exactly one grid point $p(i,j)$, for some $i,j \in \mathbf{Z}$. See Fig.~\ref{fig:f7}(left). In this case, $D$ intersects the four grid-disks $D(i,j), D(i-1,j), D(i-1,j-1), D(i,j-1)$ only and as a result, \textsc{FastCover} places at most four of these grid-disks to cover the points $D \cap P$.

	
	Now assume that $D$ contains exactly two grid points. Without loss of any generality, let the two grid points be $p(i,j)$ and $p(i,j+1)$, for some $i,j \in \mathbf{Z}$. Furthermore, we safely assume that the center of $D$ is in $\sigma(i,j)$. See Fig.~\ref{fig:f7}(right). The case where its center is in $\sigma(i-1,j)$ is analogous and thus omitted. Let $\mathcal{D}:=\{D(s,t): (s,t) \in  \{i-1,i,i+1\} \times \{j-1,j,j+1\}\}$. Since the two grid points $p(i+1,j)$ and $p(i+1,j+1)$ are not in $D$, it follows that $D \cap D(i+1,j+1) = \emptyset$ and $D \cap D(i+1,j-1) = \emptyset$. Leaving aside these two disks, $D$ can intersect with at most seven grid-disks in $\mathcal{D} \setminus \{D(i+1,j+1),D(i+1,j-1)\}$. As a result, in this case, \textsc{FastCover} will place at most seven disks to cover the points $D \cap P$. 
	
	In the final case, we assume that $D$ contains exactly four grid points $p(i,j),p(i+1,j),p(i+1,j+1),p(i,j+1)$, for some $i,j \in \mathbf{Z}$. See Fig.~\ref{fig:f11}(left). Clearly, in this case, $D$ itself is a grid-disk and  $D=D(i,j)$. It is enough to consider the nine disks $\mathcal{D}:=\{D(s,t): (s,t) \in  \{i-1,i,i+1\} \times \{j-1,j,j+1\}\}$ in this case. Among these nine disks, $D(i-1,j-1)$ will not be used just to cover a point in $D\cap P$ since in that case the only possible point $p(i,j) \in D\cap P$ also belongs to $\sigma(i,j)$. Clearly, in this case the disk $D(i,j)$ will be placed to cover the point. Similarly, the disks $D(i+1,j-1)$ and $D(i-1,j+1)$ will not be placed just to cover $p(i+1,j-1)$ and $p(i,j+1)$, respectively. This rules out three disks in $\mathcal{D}$. Thus, in this case, \textsc{FastCover} will place at most six disks to cover the points in $D \cap P$.

	Since \textsc{FastCover} will at most $7$ disks to cover the points $D \cap P$, we conclude that it gives $7$-approximation. 
	
	

	For every point in $P$, we perform exactly one look-up  and at most one insertion in $\mathcal{H}$, each taking $O(1)$ expected  time; refer to Algorithm~\ref{alg: GHS}. Hence, \textsc{FastCover} runs in $O(n)$ expected time.	Note that we insert a disk into $\mathcal{H}$ only when it belongs to the solution generated by the algorithm. This implies \textsc{FastCover} needs additional $O(s)$ space, where $s$ denotes the number of disks placed by the algorithm. 
	
	\begin{figure}[H]
		\centering
		\includegraphics[scale=0.7]{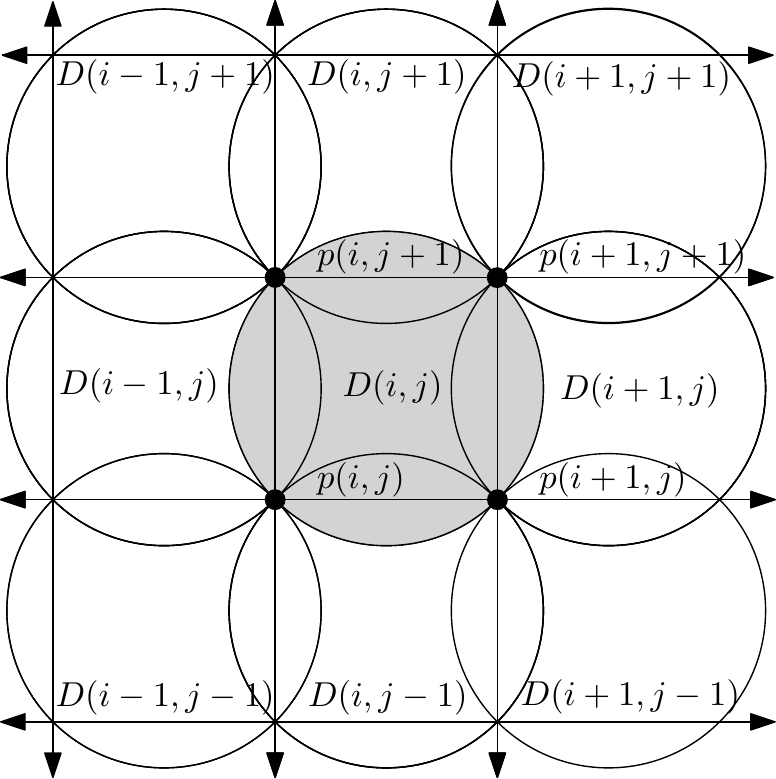} ~~~~~~~~~~~~~~~		\includegraphics[scale=0.7]{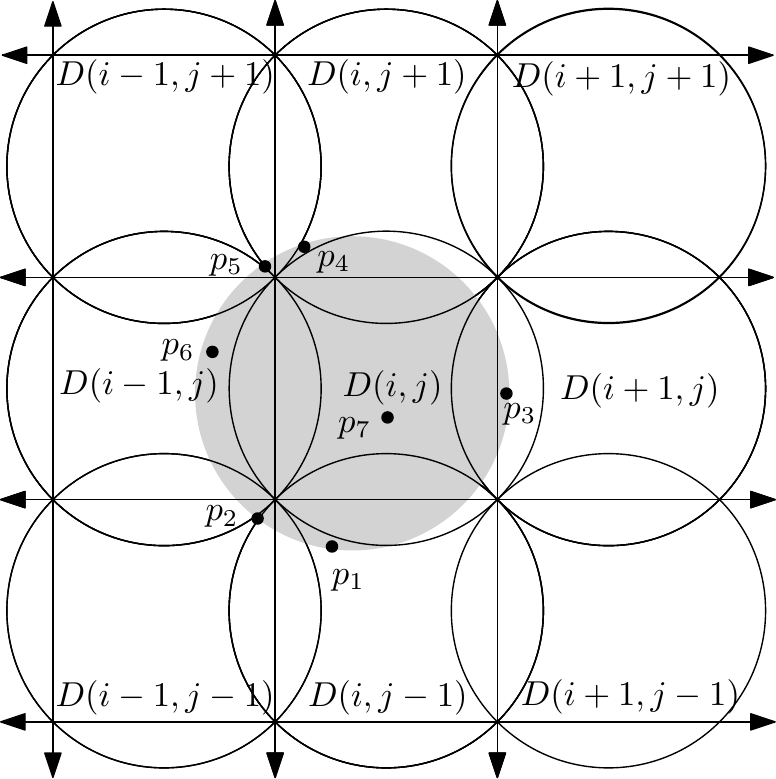}
		\caption{Left: $D$ contains exactly four grid points. Right: A sequence of seven points $p_1,\ldots,p_7$ for which \textsc{FastCover} places exactly seven disks but they can covered optimally using a single unit disk.}
		\label{fig:f11}
	\end{figure}

	Next, we present a sequence of points $Q=\{p_1,\ldots,p_7\}$ for which \textsc{FastCover} places exactly seven disks but an optimal algorithm will place exactly one disk to cover them. 
	Refer to Fig.~\ref{fig:f11}(right). Our algorithm places the seven disks: $D(i,j-1), D(i-1,j-1), D(i+1,j), D(i,j+1), D(i-1,j+1), D(i-1,j), D(i,j)$. 
	
	Given an integer $n \geq 1$, consider $n$ copies of $Q$ and place them sufficiently apart. Let these $n$ copies be $Q_1,\ldots,Q_n$, where $Q_1=Q$.  For $2 \leq i \leq n$, $Q_i$ is obtained from the $Q_{i-1}$ by adding $3$ to the $x$-coordinate of every point in $Q_{i-1}$. In this case, the $7n$ points in $\cup_{i=1}^n Q_i$ can be covered optimally using $n$ disks, but \textsc{Fast-Cover} will place exactly $7n$ disks to cover the $n$ points.
\end{proof}

\paragraph{Remark.} If the bounding-box of $P$ is known in advance and sufficient space is available, \textsc{FastCover} can be implemented to run in $O(n)$ worst-case time using a matrix for storing the disk centers. In our experiments, we have assumed that the bounding box is unknown.

 Next, we present two heuristics which are capable of improving the solutions computed by \textsc{FastCover} by decreasing the number of disks placed while still covering $P$.  
 


\paragraph{Heuristic 1.} We observe that if $p \in \sigma(i,j)$, then $p$ may also be  covered by one of the four adjacent grid-disks $N:=D(i,j+1),S:=D(i,j-1),E:=D(i+1,j),W:=D(i-1,j)$; refer to Fig.~\ref{fig:f4}. So, if $D(i,j)$ is already placed before, we do not take any action. Otherwise, we check if $p$ is covered by any one of the above four neighboring disks placed before. If not, we place $D(i,j)$. We find  this simple heuristic to be effective in reducing cover sizes in practice.


Now consider the axis-parallel square $\alpha(i,j)$ that lies inside $\sigma(i,j)$ and touches the four grid-disks $N,S,E,W$. If $p$ lies in the interior of $\alpha(i,j)$, one can safely conclude that $p$ is not covered by any of the four disks $N,S,E,W$. Note that this simple checking does not require any distance calculation. 

Let $d$ be the distance between $\alpha(i,j)$ and the boundary of $\sigma(i,j)$. Observe that $d= |BC|=|AB| = |AO| - |BO| =  1-(\sqrt{2}/2).$ 

Before verifying whether $p \in E$ using a distance calculation, we first check if $p_x \geq \sqrt{2}(i+1.5)-1$ since the right boundary of $\alpha(i,j)$ has the $x$-coordinate $\sqrt{2}(i+1) - d = \sqrt{2}(i+1) - (1-\sqrt{2}/2)= \sqrt{2}(i+1.5)-1$. 
If not, we can safely conclude that $p \notin E$. 
Similarly, before checking whether $p\in N$,  we first verify if $p_y \geq \sqrt{2}(j+1.5)-1$. For the disks $W$ and $S$, we use the conditions $p_x \leq \sqrt{2}(i-0.5)+1$ and $p_y \leq \sqrt{2}(j-0.5)+1$, respectively. 
These comparisons along with at most four checks to verify if $p$ is covered by one of the four neighboring disks help in reducing the number of disks placed in practice without slowing down the algorithm by much.

\begin{figure}[H]
	\centering
	\includegraphics[scale=0.7]{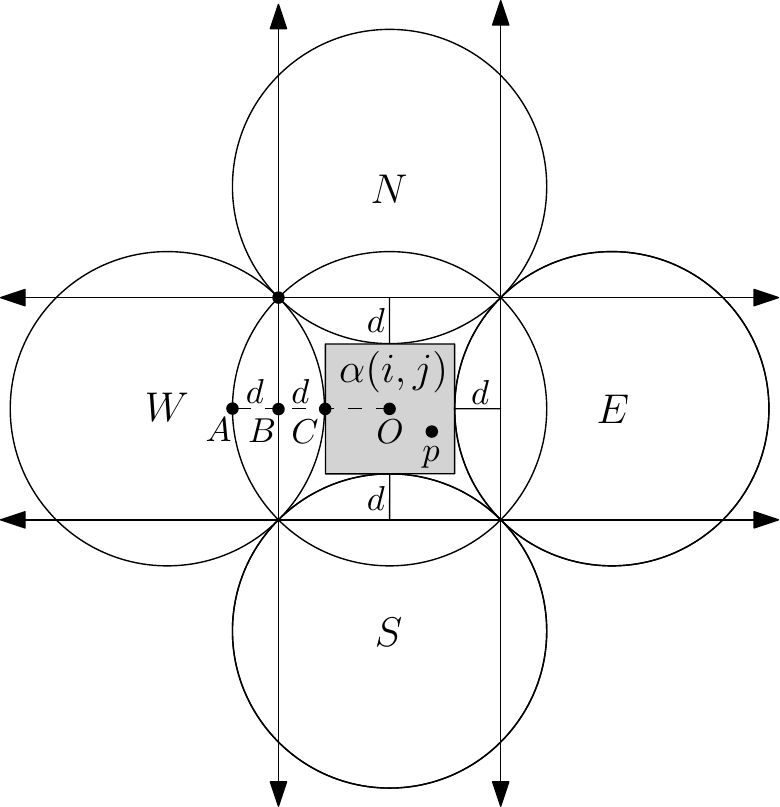}
	\caption{If $p \in \alpha(i,j)$ (shown in gray), then $p$ is not covered by any disk in $\{N,S,E,W\}$; $d = 1 - \frac{\sqrt{2}}{2} $.}
	\label{fig:f4}
\end{figure}


We propose an improved version of \textsc{FastCover} named \textsc{FastCover}\texttt{+} that includes the Heuristic 1. Refer to Algorithm~\ref{alg: fastcover+} for a pseudocode of \textsc{FastCover}\texttt{+}. 

\paragraph{Heuristic 2.} This  heuristic tries to lower down the number of disks placed by \textsc{FastCover} or \textsc{FastCover}\texttt{+} using a coalescing technique. If the points covered by two adjacent grid-disks $D(i,j),D(k,\ell)$ where $|i-k|\leq 1, |j-\ell|\leq 1$, can be covered by a single disk $D$, then we eliminate $D(i,j),D(k,\ell)$ from the solution and include $D$ instead. The pairs of grid-disks are chosen arbitrarily for coalescing. Now, the question remains to be answered is, how to confirm the existence of such a disk $D$. Certainly, one way is to find the minimum enclosing disk $D$ of the points covered by the disks $D(i,j),D(k,\ell)$ and check if its radius is at most a unit. But, every algorithm that finds a minimum enclosing disk of an $n$-element pointset  runs in $\Omega(n)$ time. 

We propose a  constant-time approximate method using bounding-boxes. For every grid-disk $D(i,j)$ placed, we maintain a bounding-box $B(i,j)$ of the points covered by $D(i,j)$. 
Now, for every point $p \in P$,  a grid-disk $D(i,j)$ is always found that covers $p$. 
If $p \in B(i,j)$, then $B(i,j)$ remains unaltered. Otherwise, $B(i,j)$ is updated to include $p$ in it. This update can be done in $O(1)$ time.  

The respective bounding-boxes $B(i,j), B(k,\ell)$ of the disks $D(i,j),D(k,\ell)$ can be used to determine in $O(1)$ time if they can be coalesced into one unit disk $D$. We compute the bounding-box $B$ by taking the union of $B(i,j)$ and  $B(k,\ell)$.  This union can be easily computed in $O(1)$ time by considering the maximum and minimum $x,y$-coordinates of the two bounding-boxes.

If the diagonal of $B$ has length at most $2$, we report that there is a disk $D$ that covers the points in $P \cap (D(i,j) \cup D(k,\ell))$. We use the center of $B$ as the center of $D$. Refer to Fig.~\ref{fig:f12} for an illustration.

\begin{algorithm}[H]
	\caption{\textsc{: \textsc{FastCover}\texttt{+}($P$)}} 
	\label{alg: fastcover+} 
	
	\begin{algorithmic}[1] 
		\STATE \texttt{$\mathcal{H}$} $\leftarrow \emptyset$; 
		 $\textsc{Disk-Centers} \leftarrow \emptyset$;
		\FOR{$p \in P$}
		\STATE  $i\leftarrow \lfloor p_x/\sqrt{2}\rfloor$;  $j\leftarrow \lfloor p_y/\sqrt{2}\rfloor$;
		
		\IF{$(i,j) \in \mathcal{H}$}
		\STATE update $B(i,j)$ using $p$; \COMMENT{$p$ is already covered by $D(i,j)$}
		
			\ELSIF{$p_x \geq \sqrt{2}(i+1.5) -1$ \textbf{and} $({i}+1,{j}) \in\mathcal{H}$ \textbf{and} \\\quad $\texttt{distance}(p, (\sqrt{2}({i}+1)+\frac{1}{\sqrt{2}},\sqrt{2}{j}+\frac{1}{\sqrt{2}})) \leq 1$}
			\STATE \textbf{continue}; \COMMENT{$p$ is covered by the grid-disk $E$ placed before} 

			\ELSIF{$p_x \leq \sqrt{2}(i-0.5)+1$ \textbf{and} $({i}-1,{j}) \in
				\mathcal{H}$ \textbf{and} \\\quad$\texttt{distance}(p, (\sqrt{2}({i}-1)+\frac{1}{\sqrt{2}},\sqrt{2}{j}+\frac{1}{\sqrt{2}})) \leq 1$}
			\STATE  \textbf{continue}; \COMMENT{$p$ is covered by the grid-disk $W$ placed before}

			\ELSIF{$p_y \geq \sqrt{2}(j+1.5) -1$ \textbf{and} $({i},{j+1}) \in
				\mathcal{H}$ \textbf{and}\\\quad $\texttt{distance}(p, (\sqrt{2}{i}+\frac{1}{\sqrt{2}},\sqrt{2}{(j+1)}+\frac{1}{\sqrt{2}})) \leq 1$}
			\STATE   \textbf{continue}; \COMMENT{$p$ is covered by the grid-disk $N$ placed before}
		
			\ELSIF{$p_y \leq \sqrt{2}(j-0.5)+1$ \textbf{and} $({i},{j-1}) \in
				\mathcal{H}$ \textbf{and} \\\quad$\texttt{distance}(p, (\sqrt{2}{i}+\frac{1}{\sqrt{2}},\sqrt{2}{(j-1)}+\frac{1}{\sqrt{2}})) \leq 1$}
			\STATE  \textbf{continue}; \COMMENT{$p$ is covered by the grid-disk $S$ placed before}
		
		\ELSE 
		\STATE insert $(i,j)$ into $\mathcal{H}$ and  $(\sqrt{2}i + \frac{1}{\sqrt{2}}, \sqrt{2}j + \frac{1}{\sqrt{2}})$ into \textsc{Disk-Centers};
		\ENDIF
		\ENDFOR
		\STATE \textbf{return} \textsc{Disk-Centers};
	\end{algorithmic}
\end{algorithm}

\begin{figure}[ht]
	\centering
	\includegraphics[scale=0.7]{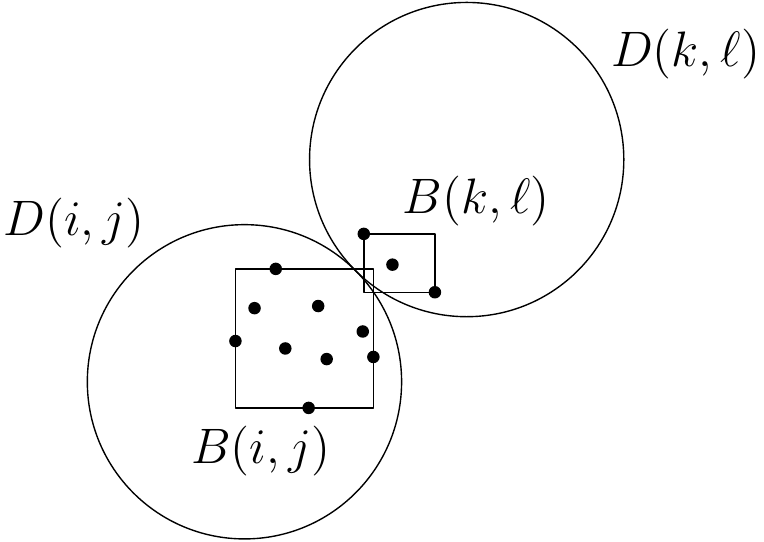} ~~~~~~~~~	\includegraphics[scale=0.7]{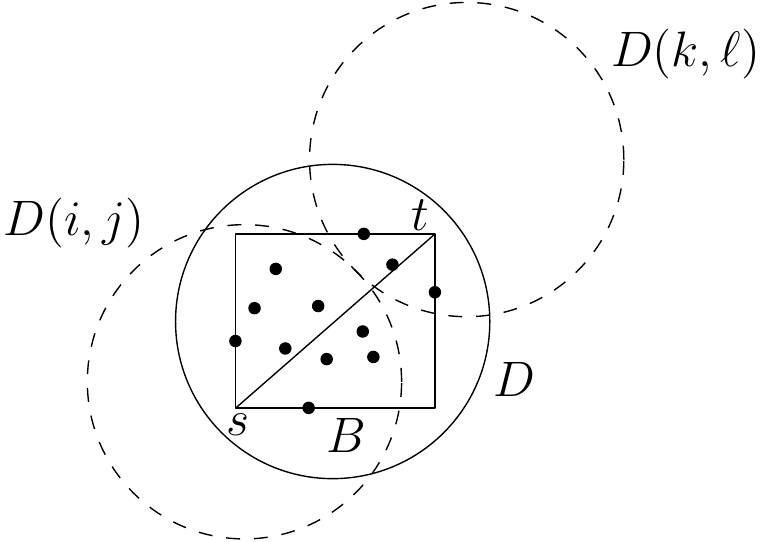} 
	\caption{Left: The case is shown where $k=i+1, \ell=j+1$. Right: The bounding-boxes $B(i,j),B(k,\ell)$ are merged to form $B$. In this case, the diagonal-length of $B$ is $|st|\leq 2$ and hence the unit disk $D$, centered at the midpoint of $st$, is used to cover the points in $B$. As a result, the disks $D(i,j),D(k,\ell)$ are eliminated from the final solution and $D$ is included instead.}
	\label{fig:f12}
\end{figure}

In this heuristic, every grid-disk is considered at most eight times. Since searching and deletion in $\mathcal{H}$ takes $O(1)$ expected time, this heuristic can be implemented to run in $O(n)$ expected time.

Now we present another improved version of \textsc{FastCover} named \textsc{FastCover}\texttt{++} that includes both Heuristics 1 and 2. See Algorithm~\ref{alg: GHS} for a pseudocode of \textsc{FastCover}\texttt{++}.

\begin{algorithm}[H]
	\caption{\textsc{: \textsc{FastCover}\texttt{++}($P$)}} 
	\label{alg: GHS} 
	
	\begin{algorithmic}[1] 
		\STATE \texttt{$\mathcal{H}$} $\leftarrow \emptyset$; 
						 $\textsc{Disk-Centers} \leftarrow \emptyset$;

		\FOR{$p \in P$}
		\STATE  $i\leftarrow \lfloor p_x/\sqrt{2}\rfloor$;  $j\leftarrow \lfloor p_y/\sqrt{2}\rfloor$;
		
		\IF{$(i,j) \in \mathcal{H}$}
		\STATE update $B(i,j)$ using $p$; \COMMENT{$p$ is already covered by $D(i,j)$}
		
		\ELSIF{$p_x \geq \sqrt{2}(i+1.5) -1$ \textbf{and} $({i}+1,{j}) \in\mathcal{H}$ \textbf{and} \\\quad $\texttt{distance}(p, (\sqrt{2}({i}+1)+\frac{1}{\sqrt{2}},\sqrt{2}{j}+\frac{1}{\sqrt{2}})) \leq 1$}
		\STATE update $B(i+1,j)$ using $p$; \COMMENT{$p$ is covered by the grid-disk $E$ placed before} 

		\ELSIF{$p_x \leq \sqrt{2}(i-0.5)+1$ \textbf{and} $({i}-1,{j}) \in
			\mathcal{H}$ \textbf{and} \\\quad$\texttt{distance}(p, (\sqrt{2}({i}-1)+\frac{1}{\sqrt{2}},\sqrt{2}{j}+\frac{1}{\sqrt{2}})) \leq 1$}
		\STATE  update $B(i-1,j)$ using $p$; \COMMENT{$p$ is covered by the grid-disk $W$ placed before}

		\ELSIF{$p_y \geq \sqrt{2}(j+1.5) -1$ \textbf{and} $({i},{j+1}) \in
			\mathcal{H}$ \textbf{and}\\\quad $\texttt{distance}(p, (\sqrt{2}{i}+\frac{1}{\sqrt{2}},\sqrt{2}{(j+1)}+\frac{1}{\sqrt{2}})) \leq 1$}
		\STATE   update $B(i,j+1)$ using $p$; \COMMENT{$p$ is covered by the grid-disk $N$ placed before}
		
		\ELSIF{$p_y \leq \sqrt{2}(j-0.5)+1$ \textbf{and} $({i},{j-1}) \in
			\mathcal{H}$ \textbf{and} \\\quad$\texttt{distance}(p, (\sqrt{2}{i}+\frac{1}{\sqrt{2}},\sqrt{2}{(j-1)}+\frac{1}{\sqrt{2}})) \leq 1$}
		\STATE   update $B(i,j-1)$ using $p$; \COMMENT{$p$ is covered by the grid-disk $S$ placed before}
		
		\ELSE 
		\STATE insert $(i,j)$ into $\mathcal{H}$ and initialize $B(i,j)$ using $p$;
		\ENDIF
		\ENDFOR

		\WHILE{there is a grid-disk $(i,j) \in \mathcal{H}$ that is not considered yet}
		\IF{there is a grid disk $(k,\ell) \in \mathcal{H}$ such that $|i-k|\leq 1, |j-\ell|\leq 1$ and the diagonal-length of the bounding-box $B:=B(i,j)\cup B(k,\ell)$ is at most $2$}
		\STATE remove $(i,j)$ and $(k,\ell)$ from $\mathcal{H}$ and 
	 add the center of $B$ to \textsc{Disk-Centers};
		\ENDIF
		\ENDWHILE
		
		\FOR{every grid-disk $(i,j) \in \mathcal{H}$}
		\STATE insert $(\sqrt{2}i + \frac{1}{\sqrt{2}}, \sqrt{2}j + \frac{1}{\sqrt{2}})$ into \textsc{Disk-Centers};
		\ENDFOR
		
		\STATE \textbf{return} \textsc{Disk-Centers};
	\end{algorithmic}
\end{algorithm}


\section{Engineering and experiments}\label{sec:exp}

In this section, we present our experimental results. 
The algorithms have been implemented in GNU C\texttt{++}$17$ using the CGAL library~\cite{cgal:eb-21b}. The machine used for our experiments is equipped with a {Ryzen 5 1600 (3.2 GHz) processor, $24$ GB of main memory, and runs Ubuntu Linux 20.04 LTS}.  During compilation, the \texttt{g++} compiler was invoked with \texttt{-O3} optimization flag for fastest real-world execution times. 
\paragraph{Implementation details.} For better real-world performance, containers and functions from the C\texttt{++} STL (Standard Template Library) are used wherever needed. The \texttt{CGAL::Cartesian<double>} from CGAL is used for geometric computations.  We have tried our best to tune our codes to run fast. For instance, constant expressions used in the code have been pre-calculated and stored in variables to avoid repeated calculations. For measuring time, we have used \texttt{std::chrono::high\_resolution\_clock}. Wherever needed, for computing distance between two points we have used the \texttt{CGAL::squared\_distance} function.

\begin{itemize}\itemsep0pt
	\item G-1991: Instead of partitioning the horizontal strips into two groups as stated in the algorithm,  we have used one group and then processed the  strips sequentially for faster speed. Since ordering is needed, \texttt{std::map} is used for storing the strips along with the points.
	
	\item CCFM-1997: As evident from its pseudocode, this algorithm needs nearest neighbor searching for every point. Although CGAL has support for such queries, we have used a grid-based approach akin to \textsc{FastCover} for speed. We maintain the non-empty grid cells using a \texttt{std::unordered\_map} with \texttt{boost::hash<std::pair<int,int>>} as the hash function for fast real-world speed. For every cell, we maintain a list of disk centers inside it. In this way the nearest neighbor query for the current point under consideration can be executed fast since inside the cell in which the point lies, there can be a constant number of disk centers placed by CCFM-1997. The same observation holds for the neighboring cells. For every point, we need to probe into at most nine cells and consequently, nearest neighbor queries implemented in this fashion run very fast for this algorithm. 

	\item LL-2014 and LL-2014-1P: We have used \texttt{std::sort} for sorting. Recall that LL-2014 uses six passes for computing covers. In our experiments, we find that there is barely any difference in cover sizes when one pass is used as opposed to six passes.  Obviously, the one-pass version is much faster too. We named this one pass version LL-2014-1P. 
	
	\item BLMS-2017: As stated in Section~\ref{sec:algorithms}, this algorithm may place empty disks sometimes. In order to eliminate such disks from the solution, we keep track of the empty disks. As proposed by the authors of BLMS-2017, we have used a binary search tree to implement this algorithm. This tree is implemented using \texttt{std::set}.
	
	\item DGT-2018: Similar to CCFM-1997, this algorithm is also dependent on nearest neighbor queries. As such we have used the grid-based setup as used in CCFM-1997. 
	
	\item \textsc{FastCover} and \textsc{FastCover}\texttt{+}: We have used \texttt{std::unordered\_set} to implement the required hash-table for this algorithm with \texttt{boost::hash<std::pair<int,int>>} as the hash function in order to maintain the set of placed grid-disks.

	\item \textsc{FastCover}\texttt{++}: In this variation of \textsc{FastCover} we have used \texttt{std::unordered\_map} with \texttt{boost::hash<std::pair<int,int>>} as the hash function to maintain the set of placed grid-disks and their corresponding bounding-boxes. 
\end{itemize}

In our experiments, we have used both synthetic and real-world pointsets. The execution times of most UDC algorithms vastly depend on the density of pointsets. Clearly, covers of high-density pointsets always have small sizes. Interestingly, this enhances the speed of most UDC algorithms since they consult the set of disks already placed to verify if the current point under consideration is already covered by at least one of those disks. This motivates us to use pointsets having varied densities.

\paragraph{Synthetic pointsets.} For generating synthetic pointsets, we have used the following three random pointset generators: \texttt{CGAL::Random\_points\_in\_square\_2} (generates pointsets inside a square), \texttt{CGAL::Random\_points\_in\_disc\_2} (generates pointsets inside a disk), and \texttt{CGAL::random\_convex\_set\_2} (generates convex pointsets inside a square). Besides, we also have used pointsets drawn from an annulus. For generating points inside annulus, we have used \texttt{std::default\_random\_engine} and \texttt{std::uniform\_real\_distribution<double>} to generate point coordinates.   In out experiments, the sizes of the pointsets are in the order of millions. They range from $1$ to $10$ million. For every value of $n$, we have drawn $5$ samples for our experiments. The reported times and cover sizes are averaged over five samples. The plots and tables for synthetic pointsets have been consolidated and moved to Section~\ref{sec:plots} for an easy reference. 

\begin{itemize}\itemsep0pt
\item \textbf{Points drawn from a square.} For this class of pointsets, we have used four bounding boxes having areas $10^7, 10^6, 10^5, 10^4$ to have pointsets of varying densities. Refer to the Figs.~\ref{fig:PIS-10-7},~\ref{fig:PIS-10-6},~\ref{fig:PIS-10-5},and~\ref{fig:PIS-10-4}. In Fig.~\ref{fig:PIS-10-7}, the point densities are in the range $0.1,0.2,\ldots,1$; in Fig.~\ref{fig:PIS-10-6}, they are in the range $1,2,\ldots,10$; in Fig.~\ref{fig:PIS-10-5}, the range is  $10,20,\ldots,100$, and  in Fig.~\ref{fig:PIS-10-4}, the range is  $100,200,\ldots,1000$. Thus, our pointsets have densities that vary from as low as $0.1$ to as high as $1000$. 

As evident from Figs.~\ref{fig:PIS-10-7},~\ref{fig:PIS-10-6},~\ref{fig:PIS-10-5},and~\ref{fig:PIS-10-4}, LL-2014 turned out to be the slowest of all. The main reason for this slowdown is the number of passes it makes. It makes $6$ passes to bring down the approximation factor to $25/6$ from $5$. However, we find that $6$ passes are not that helpful in reducing the number of disks placed. Indeed,  as can be seen in our experimental data, there is a negligible difference between LL-2014 and LL-2014-1P (the one-pass version of LL-2014) when it comes to cover size.  On the other hand, in terms of speed, LL-2014-1P is much faster than LL-2014 since it is making just one pass. LL-2014-1P did very well in generating low-sized covers. In fact, its performance is one of the best overall in this regard. In Fig.~\ref{fig:PIS-10-6}, we find it is way ahead of others in placing fewer disks. In terms of speed, LL-2014-1P remained very competitive everywhere.

BLMS-2017 and CCFM-1997 are some of the slowest in this case.  They are also lagging behind others in terms of the disks they place. 
Interestingly, G-1991 performed really well both in terms of speed and cover size, especially for high-density pointsets; see Figs.~\ref{fig:PIS-10-5} and~\ref{fig:PIS-10-4}. This is  surprising because it has an approximation factor of $8$, the highest among the algorithms we have engineered. In Fig.~\ref{fig:PIS-10-4}, we find that it placed the least number of disks while being very fast. We find DGT-2018 to be less efficient than most other algorithms regarding the number of disks placed. 

\textsc{FastCover} turned out to be the fastest of all because of its sheer simplicity and no use of any pre-processing such as sorting.  \textsc{FastCover}\texttt{+} (Heuristic 1 included) really helped to bring down the number of disks placed without any considerable slowdown. \textsc{FastCover}\texttt{++} with Heuristic 2 reduced the number of disks further without slowing down too much.
Refer to Fig.~\ref{fig:PIS-10-7} for instance. In most of the cases, \textsc{FastCover}\texttt{+} and \textsc{FastCover}\texttt{++} could decrease the number of disks placed considerably. When $n=10M$, \textsc{FastCover} has placed $4,323,694$ disks on average, \textsc{FastCover}\texttt{+} has placed $3,752,103$  disks ($ \approx 13\%$ reduction in cover size), and \textsc{FastCover}\texttt{++} has placed $2,794,374$ disks ($\approx 35\%$ reduction in cover size compared to \textsc{FastCover}).
The effect of these heuristics seems to vanish with the increase in point density since for high-density pointsets almost every disk placed by \textsc{FastCover} is required. In other words, for every disk that is placed, it is highly likely that there is at least one point that is covered by that disk only. Our experimental data says that the coalescing strategy used in Heuristic 2 is not slowing down the algorithm heavily. But this slowdown is much less noticeable for higher density pointsets since, with the increase in density, the number of disks maintained by the algorithm decreases. As a result, there is a lesser number of disks to be considered for coalescing. This makes \textsc{FastCover}\texttt{++} very speedy for high-density points, see Figs.~\ref{fig:PIS-10-6},~\ref{fig:PIS-10-5},~\ref{fig:PIS-10-4}.

\item \textbf{Points drawn from a disk.} Just like our previous setup (points drawn from a square), we have used disks having areas $10^7, 10^6, 10^5, 10^4$. Refer to the Figs.~\ref{fig:PID-10-7},~\ref{fig:PID-10-6},~\ref{fig:PID-10-5},and~\ref{fig:PID-10-4}. The results obtained are quite similar to the previous setup, as can be observed in the figures and therefore we omit the discussion. 

\item \textbf{Convex pointsets drawn from a square.} We have drawn random convex sets from squares whose areas are in $\{10^7, 10^6, 10^5, 10^4\}$. Interestingly, \textsc{FastCover}\texttt{++} always placed the minimum number of disks. In terms of speed, \textsc{FastCover} turned out to be the fastest but \textsc{FastCover}\texttt{++} was very close to it everywhere. To see this, refer to the Figs.~\ref{fig:Conv-10-7},~\ref{fig:Conv-10-6},~\ref{fig:Conv-10-5},and~\ref{fig:Conv-10-4}.  To our surprise, we find that unlike the previous two classes of pointsets, LL-2014-1P did not perform well in this case both in terms of speed and the number of disks placed. BLMS-2017 and CCFM-1997 turned out to be inferior both in terms of speed and cover size. G-1991 and DGT-2018 were very speedy in this case but lagged behind the \textsc{FastCover}\texttt{++} in minimizing the number of disks placed. Our heuristics 1 and 2 together could substantially decrease cover sizes for this class of pointsets too. Refer to Fig.~\ref{fig:Conv-10-7} for instance. For $n=1M$, \textsc{FastCover} has placed $8941$ disks on average. In contrast, \textsc{FastCover}\texttt{++} has placed $5964$ disks on average. This is a $\approx 33\%$ reduction in the number of disks placed. Note that the average runtime of \textsc{FastCover} in this case is $0.15$ second and that of \textsc{FastCover}\texttt{++} is $0.44$ second. The heuristic 1 alone could not achieve much for this class  of pointsets. This is evident from the number of disks placed by \textsc{FastCover} vs \textsc{FastCover}\texttt{+}.

\item \textbf{Points drawn from an annulus.} In this setup, we have fixed the radius $r_2$ of the outer circle to $10^3$ and have varied the radius $r_1$ of the inner circle. Refer to the Figs.~\ref{fig:Ann-95},~\ref{fig:Ann-75},and~\ref{fig:Ann-50}. In our experiments, $r_1 \in  \{0.95\cdot 10^3, 0.75\cdot 10^3, 0.5\cdot 10^3\}$. For this class of pointsets, LL-2014 and LL-2014-1P came out as the clear winner considering the cover sizes. G-1991 captured  the second place in this regard and was competitive in  speed. But in terms of speed, \textsc{FastCover}, \textsc{FastCover}\texttt{+}, and \textsc{FastCover}\texttt{++} turned out to be the fastest of all and remained competitive in solution quality. The Heuristics 1 and 2 did a decent job of reducing the cover sizes but they were not effective as in the cases of other three classes of pointsets. For instance, refer to Fig.~\ref{fig:Ann-95}, $n=1M$, \textsc{FastCover} has placed $155,984.80$ disks on average. In contrast, \textsc{FastCover}\texttt{+} has placed $150,294.80$ disks ($4\%$ less number of disks), and \textsc{FastCover}\texttt{++} has placed $141,845.80$ disks ($9\%$ less number of disks compared to \textsc{FastCover}) on average.
BLMS-2017, CCFM-1997, and DGT-2018 were not only slower but also turned out to be less efficient in minimizing the cover sizes. 

\end{itemize}

\paragraph{Real-world pointsets.} We have used the following eleven real-world pointsets for our experiments. \change{The main reason behind the use of such pointsets is that they do not follow the popular synthetic distributions. Hence,  experimenting with them is beneficial to see how the algorithms perform on them both in speed and the number of disks placed. Further,  the UDC algorithms can be used on real-world pointsets to tackle practical coverage problems such as facility location and tower placements in wireless networking. Hence, experiments with real-world pointsets throw light on the real-world efficacy of these algorithms.}
See Fig.~\ref{fig:realWorld} for the experimental results obtained for these pointsets. Note that these data sets have varied bounding box sizes.

\begin{itemize}\itemsep0pt
    \item \texttt{birch3}~\cite{bus2018practical}: An $100,000$-element pointset representing random sized clusters at random locations. Area of the bounding box: $8.84532 \times 10^{11}$. Point density: $1.13\times 10^{-7}$
    
   \item \texttt{monalisa}~\cite{tsp}: A $100,000$-city TSP instance representing a continuous-line drawing of the Mona Lisa. Area of the bounding box: $3.98681 \times 10^{8}$. Point density: $ 0.00025$

    \item \texttt{usa}~\cite{tsp}: A $115,475$-city TSP instance representing (nearly) all towns, villages, and cities in the United States. Area of the bounding box: $1.43141 \times 10^{9}$. Point density: $ 0.00008$
    
    \item \texttt{KDDCU2D}~\cite{bus2018practical}: An $145,751$-element pointset representing the first two dimensions of a protein data-set. Area of the bounding box: $8564.16$. Point density: $ 17.019$

    \item \texttt{europe}~\cite{bus2018practical}: An $169,308$-element pointset representing differential coordinates of the map of Europe. Area of the bounding box: $10^{12}$. Point density: $1.69308 \times 10^{-7}$
    
    \item \texttt{wildfires}\footnote{\url{https://www.kaggle.com/rtatman/188-million-us-wildfires/home}}: An $1,880,465$-element pointset representing wildfire locations in USA. Area of the bounding box: $5948.76$. Point density: $316.11$
    
    \item \texttt{world}~\cite{tsp}: A $1,904,711$-city TSP instance consisting of all locations in the world that are registered as populated cities or towns, as well as several research bases in Antarctica. Area of the bounding box: $58938.8$.  Point density: $32.32$
    
   \item \texttt{china}~\cite{bus2018practical}: An $1,636,613$-element pointset representing locations in China. Area of the bounding box: $3188.92$. Point density: $513.219$
 
     \item \texttt{nyctaxi}\footnote{\url{https://www.kaggle.com/wikunia/nyc-taxis-combined-with-dimacs/home}}: An $2,917,288$-element pointset representing NYC taxi pickup and drop-off locations.  Area of the bounding box: $1193.77$. Point density: $2443.76$
    
    \item \texttt{uber}\footnote{\url{https://www.kaggle.com/fivethirtyeight/uber-pickups-in-new-york-city}}: An $4,534,327$-element pointset representing Uber pickup locations in New York City. Area of the bounding box: $7.04065$. Point density: $644021.078$
    
    \item \texttt{hail2015}\footnote{\url{https://www.kaggle.com/noaa/severe-weather-data-inventory}}: An $10,824,080$-element pointset representing hail storm cell locations based on NEXRAD radar data obtained in 2015. Area of the bounding box: $17921.8$. Point density: $603.96$
\end{itemize}

\begin{figure}[ht]
	\centering

	\includegraphics[scale=0.66]{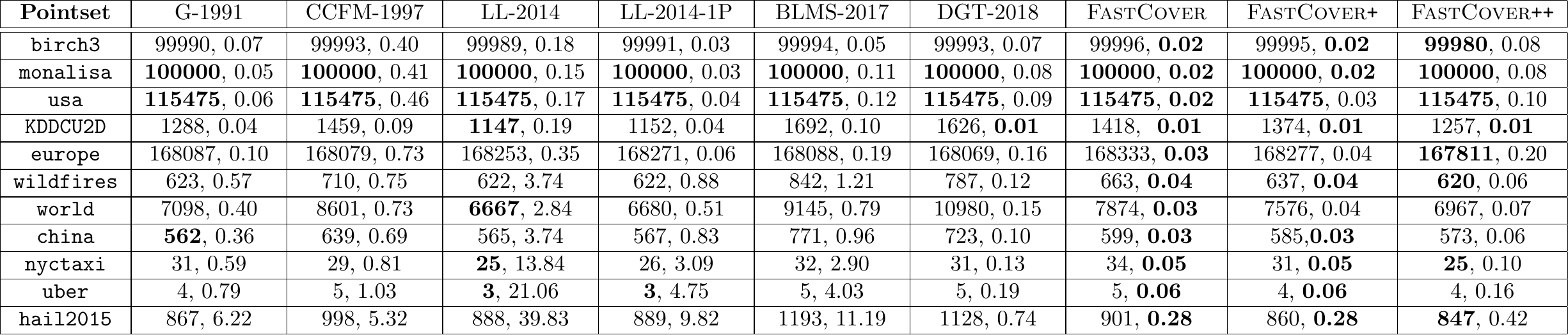}

	\caption{Experimental results for the real-world pointsets.  A pair $x,y$ in a cell denotes the number of disks placed and the running time in seconds, respectively, averaged over five runs of the same pointset. For every pointset, the smallest sized cover(s) and the fastest execution time(s) are shown in bold.}
	\label{fig:realWorld}
\end{figure}

As evident from our experimental results in Fig.~\ref{fig:realWorld} \textsc{FastCover}, \textsc{FastCover}\texttt{+}, and \textsc{FastCover}\texttt{++} performed really well on the real-world pointsets both in terms of speed and cover size. The Heuristics 1 and 2 performed reasonably well wherever possible. For the \texttt{KDDCU2D} pointset, \textsc{FastCover},  \textsc{FastCover}\texttt{+}, and \textsc{FastCover}\texttt{++} have placed $1418$, $1374$, and $1257$ disks, respectively. In this case, the Heuristic 1 could reduce cover size by approximately $3\%$ and the heuristics 1 and 2 together could reduce the size by approximately $11\%$.
For the \texttt{world} pointset, these percentages are $\approx 4$ and $11$, respectively. For these real-world pointsets, there is a negligible difference between \textsc{FastCover} and \textsc{FastCover}\texttt{++} in terms of speed. 
To our surprise, we found that for some of the pointsets such as \texttt{wildfires}, and \texttt{hail2015}, \textsc{FastCover}\texttt{++} has returned the smallest covers and quickly ran to completion. The other algorithms such as  LL-2014-1P, G-1991, BLMS-2017, CCFM-1997, DGT-2018 turned out to be reasonably fast and returned competitive solutions on all the pointsets. G-1991 has placed the lowest number of disks for the \texttt{china} pointset. As in the case of synthetic pointsets, we found LL-2014 to be slower than LL-2014-1P.  Further, LL-2014 could not reduce cover sizes substantially compared to its 1-pass version LL-2014-1P. 

\section{Our recommendations and conclusions}\label{sec:con}

If the sizes of covers are of more interest than real-world running time, we recommend using LL-2014-1P. However, we also observe that in some cases, \textsc{FastCover}\texttt{++} beats LL-2014 both in terms of speed and cover size.

If running time is much more important than cover size, we recommend using \textsc{FastCover}. Otherwise, we recommend using \textsc{FastCover}\texttt{++} since it is very fast in practice and at the same time can generate low-sized covers. 
If it is known that the input pointsets are always convex, we recommend \textsc{FastCover}\texttt{++} since in this case, it beats every other algorithm both in terms of speed and cover size.  

Overall, we find that either LL-2014-1P or \textsc{FastCover}\texttt{++} is leading when it comes to the number of disks placed. 
So, in situations, where running time is not so 
important but the cover size is, one can run both LL-2014-1P and 
\textsc{FastCover}\texttt{++} and 
take the best of the two covers.

\medspace

\noindent
\textbf{Acknowledgments.} We sincerely thank all the anonymous reviewers of this manuscript for their careful reading and their numerous insightful comments and suggestions for improvements. \change{We especially thank the anonymous reviewer who has generously shared improved implementations with us for the algorithms G-1991, BLMS-2017, DGT-2018, DGT-2018, and CCFM-1997. }


%
%

\newpage

\section{Plots and tables}\label{sec:plots}

In the interest of space, we avoid legend tables everywhere in our plots. Since the legends are used uniformly throughout this work, we present them here for an easy reference; refer to Fig.~\ref{fig:legends}. In the tables, we mark the fastest execution time and the smallest covers in bold. 

\begin{figure}[h]
	\centering
	\begin{tikzpicture}
		\begin{axis}[legend pos=south west,	hide axis, xmin=10,	xmax=0, ymin=0, ymax=0.4,legend columns=5]
			\addlegendimage{color=blue,mark=star}
			\addlegendentry{G-1991};
			
			\addlegendimage{color=red,mark=diamond*}
			\addlegendentry{CCFM-1997};
			
			\addlegendimage{color=black,mark options={fill=yellow},mark=square*}
			\addlegendentry{\textsc{LL-2014-1P}};
			
			\addlegendimage{color=orange,mark=triangle*}
			\addlegendentry{LL-2014};
			
			\addlegendimage{color=black,mark=square*}
			\addlegendentry{BLMS-2017};
			
			\addlegendimage{color=teal,mark=pentagon*}
			\addlegendentry{DGT-2018};
			
			\addlegendimage{color=black,mark options={fill=yellow},mark=otimes*}
			\addlegendentry{\textsc{FastCover}};
			
			\addlegendimage{color=black,mark options={fill=yellow},mark=oplus*}
			\addlegendentry{\textsc{FastCover\texttt{+}}};
			
			\addlegendimage{ color=purple,mark=*}
			\addlegendentry{\textsc{FastCover\texttt{++}}};
			
		\end{axis}
	\end{tikzpicture}
	\caption{The plot legends.}
	\label{fig:legends}
\end{figure}
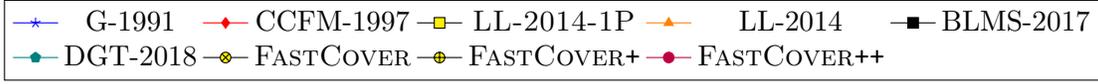

\begin{figure}[H]
	\centering
\includegraphics[scale=0.75]{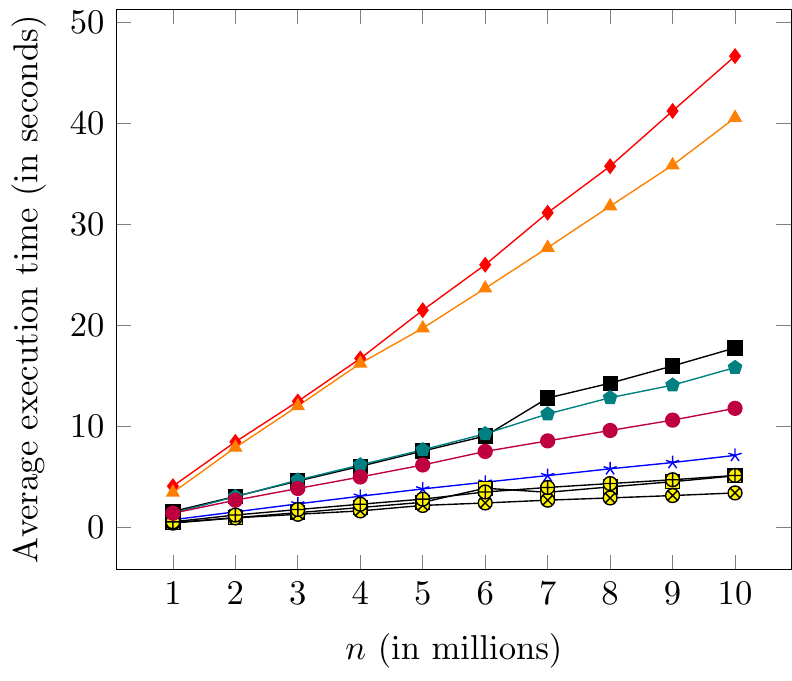} \hspace{55pt}	
\includegraphics[scale=0.75]{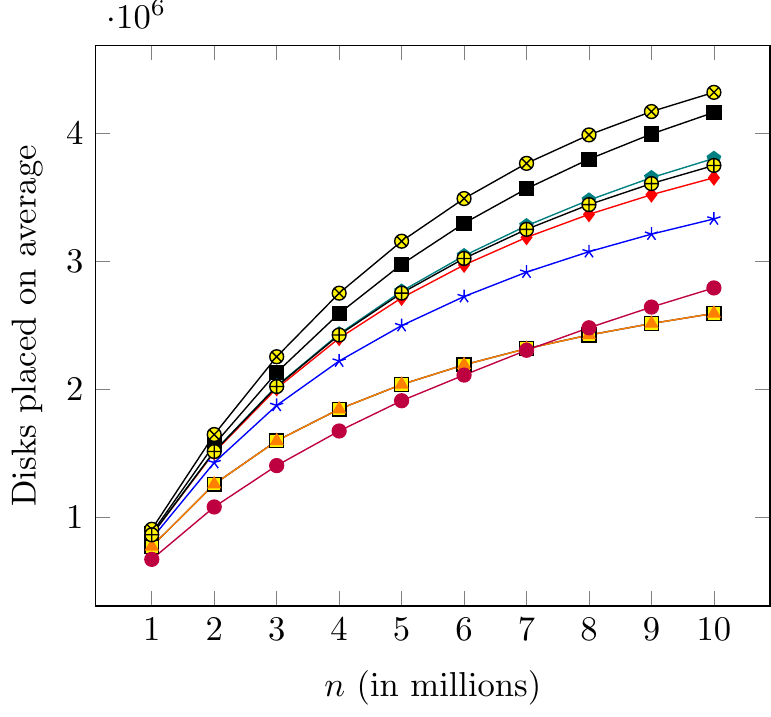}\\ \vspace*{20pt}
\includegraphics[scale=0.58]{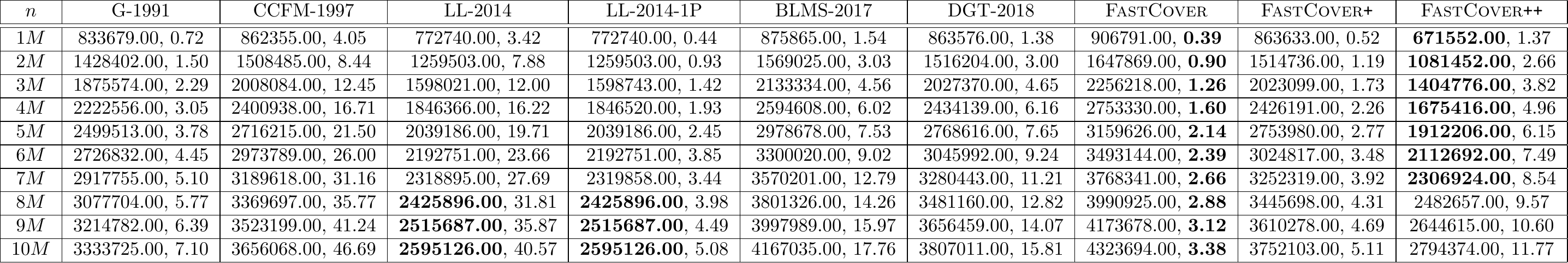}
\caption{Points drawn from a square of area $10^7$. A pair $x,y$ in a cell denotes the average number of disks placed and the average running time  in seconds, respectively.}
\label{fig:PIS-10-7}
\end{figure}

\begin{figure}[H]
	\centering
	\includegraphics[scale=0.75]{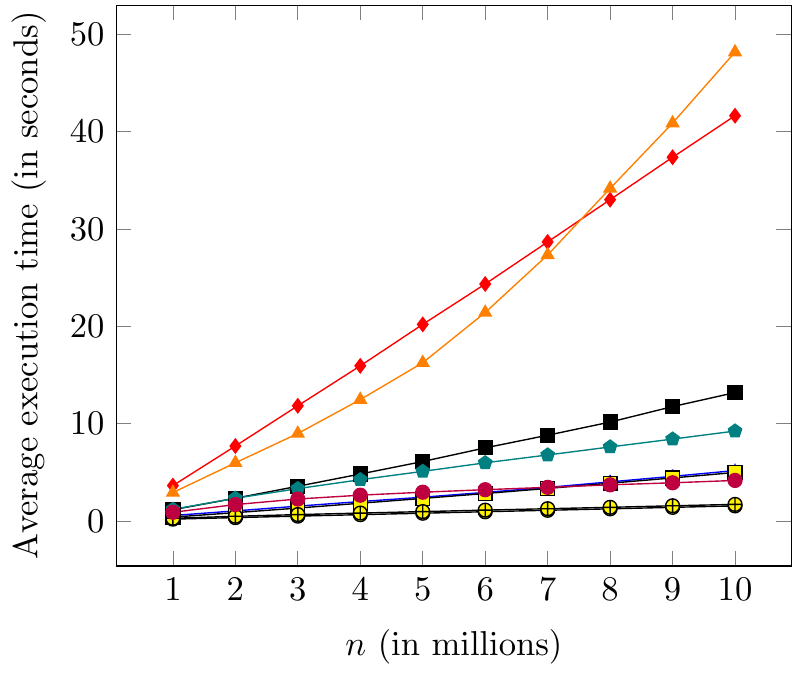} \hspace{55pt}	
	\includegraphics[scale=0.75]{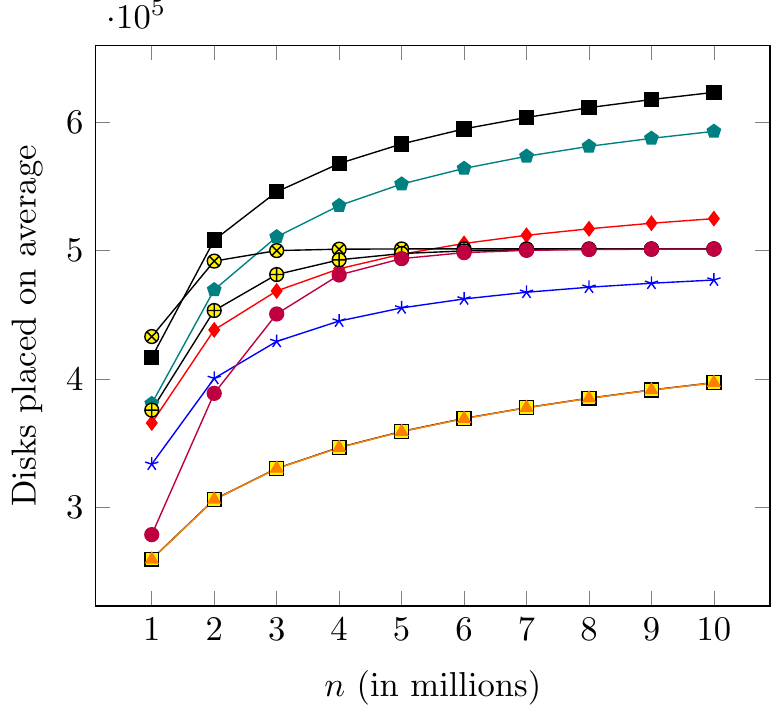}\\ \vspace*{20pt}
	\includegraphics[scale=0.58]{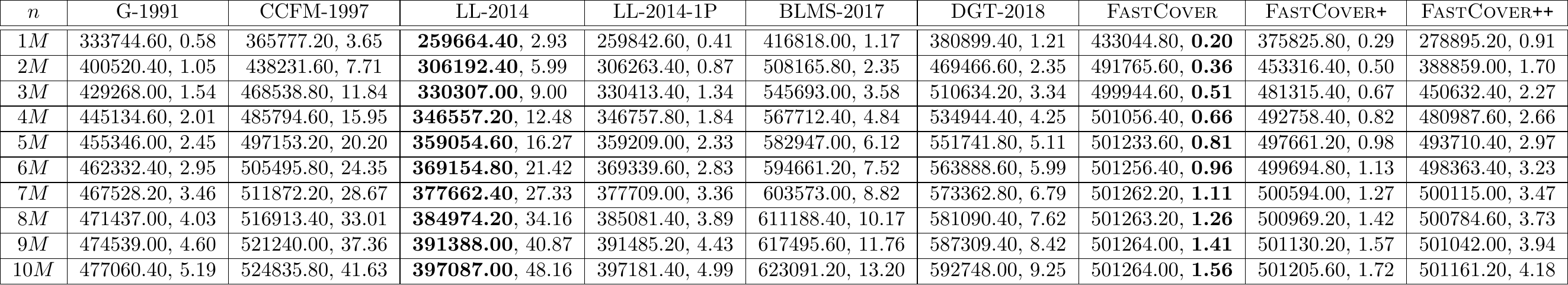}
	\caption{Points drawn from a square of area $10^6$. A pair $x,y$ in a cell denotes the average number of disks placed and the average running time  in seconds, respectively.}
	\label{fig:PIS-10-6}
\end{figure}

\begin{figure}[H]
		\centering
	\includegraphics[scale=0.75]{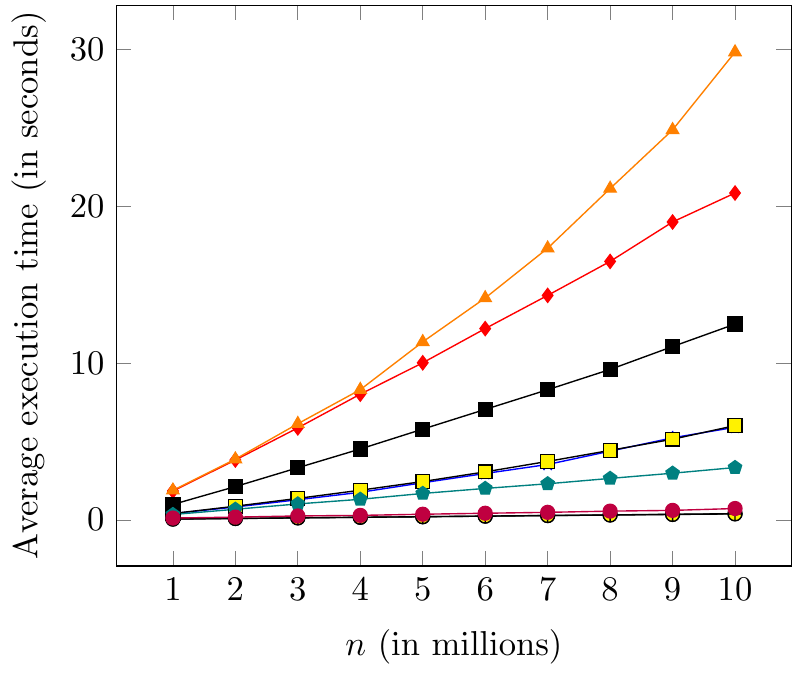} \hspace{55pt}	
	\includegraphics[scale=0.75]{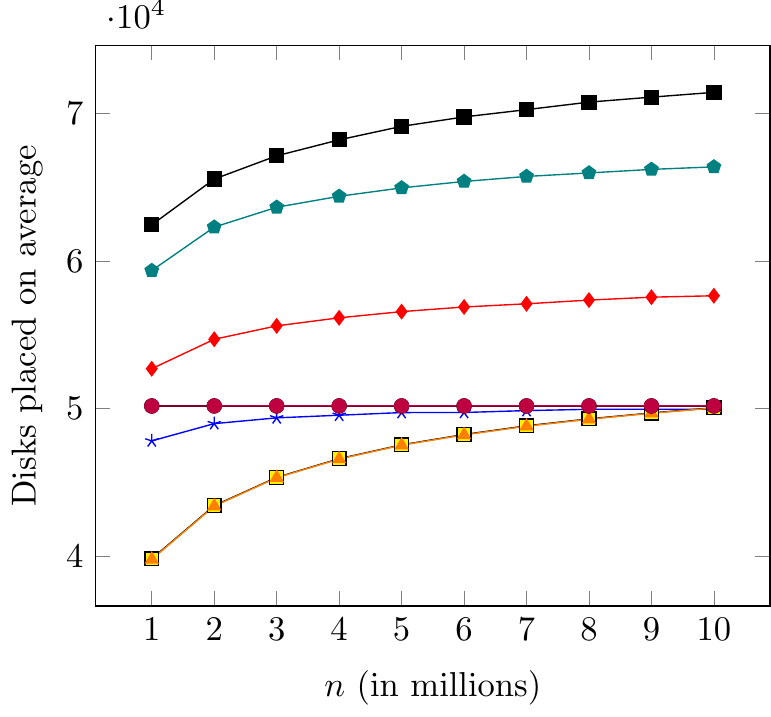}\\ \vspace*{20pt}
	\includegraphics[scale=0.58]{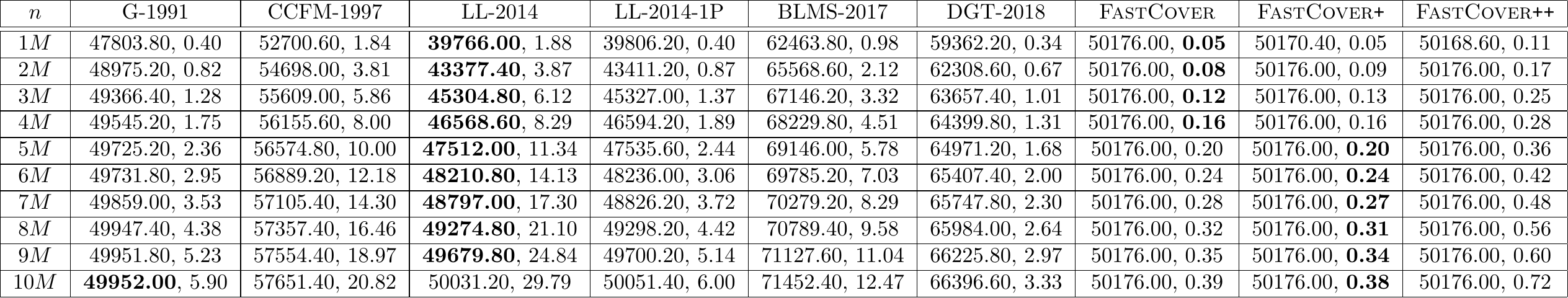}
\caption{Points drawn from a square of area $10^5$. A pair $x,y$ in a cell denotes the average number of disks placed and the average running time in seconds, respectively. }
	\label{fig:PIS-10-5}
\end{figure}

\begin{figure}[H]
	\centering
	\includegraphics[scale=0.75]{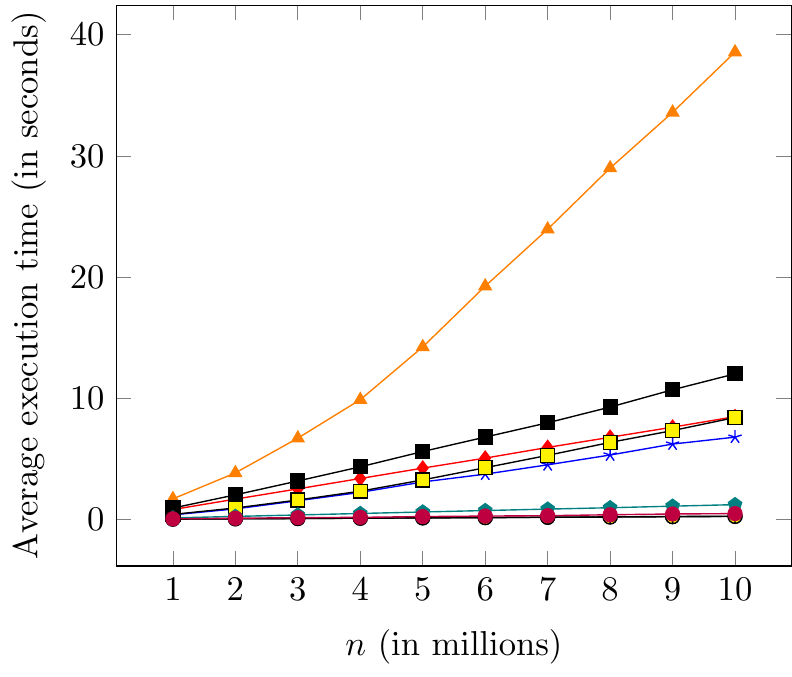} \hspace{55pt}	
	\includegraphics[scale=0.75]{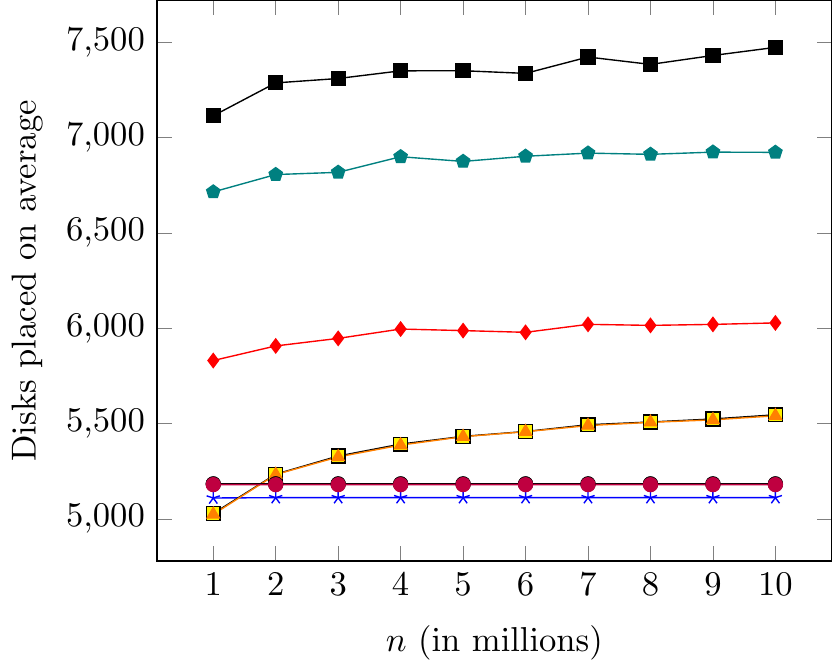}\\ \vspace*{20pt}
	\includegraphics[scale=0.58]{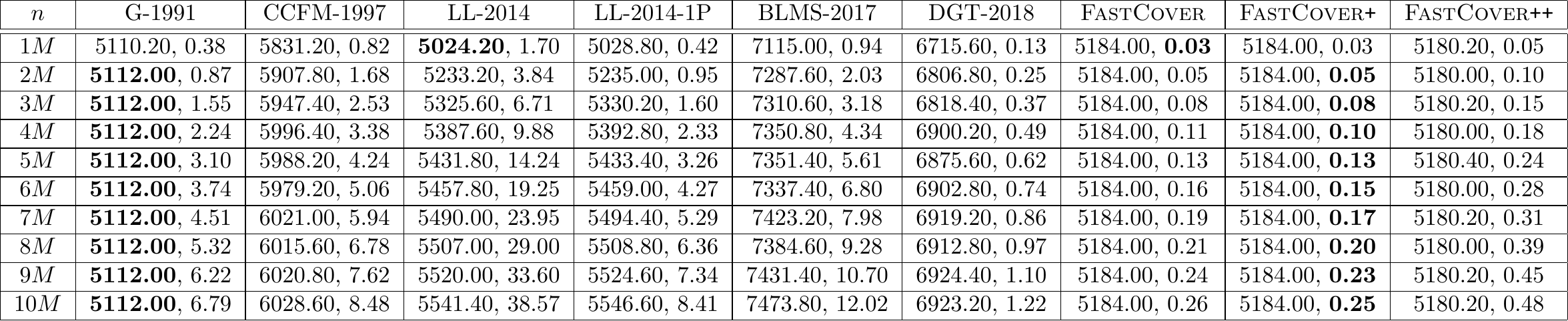}
	\caption{Points drawn from a square of area $10^4$. A pair $x,y$ in a cell denotes the average number of disks placed and the average running time in seconds, respectively. }
	\label{fig:PIS-10-4}
\end{figure}


\begin{figure}[H]
	\centering
	\includegraphics[scale=0.75]{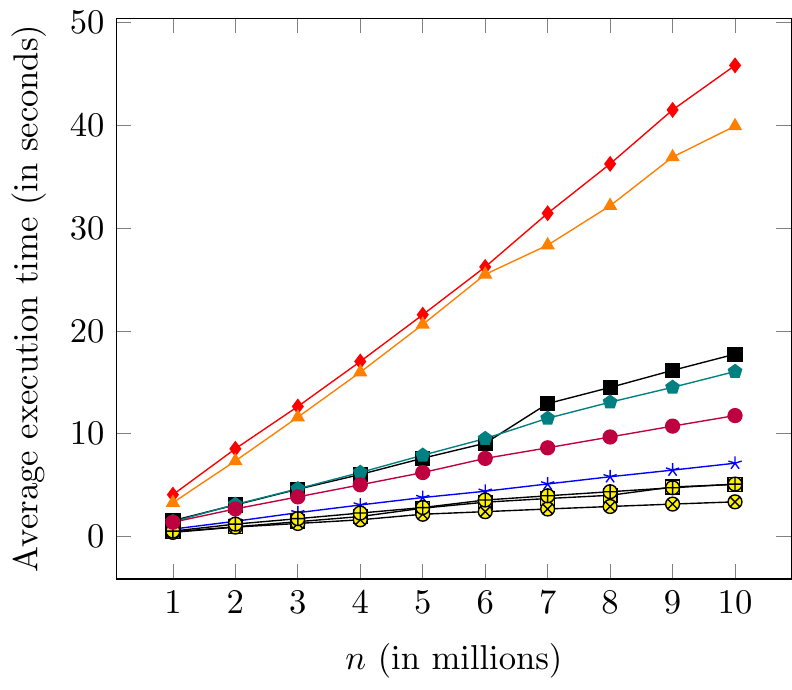} \hspace{55pt}	
	\includegraphics[scale=0.75]{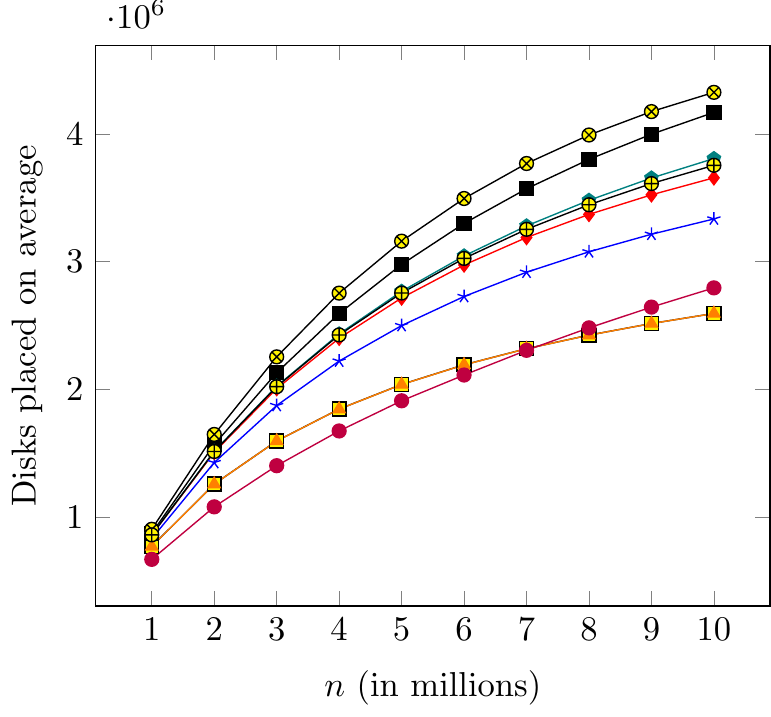}\\ \vspace*{20pt}
	\includegraphics[scale=0.58]{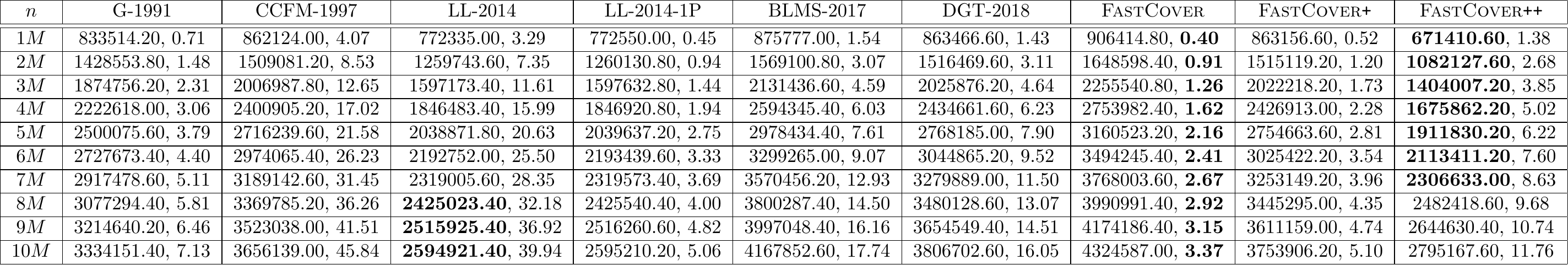}
	\caption{Points drawn from a disk of area $10^7$. A pair $x,y$ in a cell denotes the average number of disks placed and the average running time  in seconds, respectively.}
	\label{fig:PID-10-7}
\end{figure}

\begin{figure}[H]
	\centering
	\includegraphics[scale=0.75]{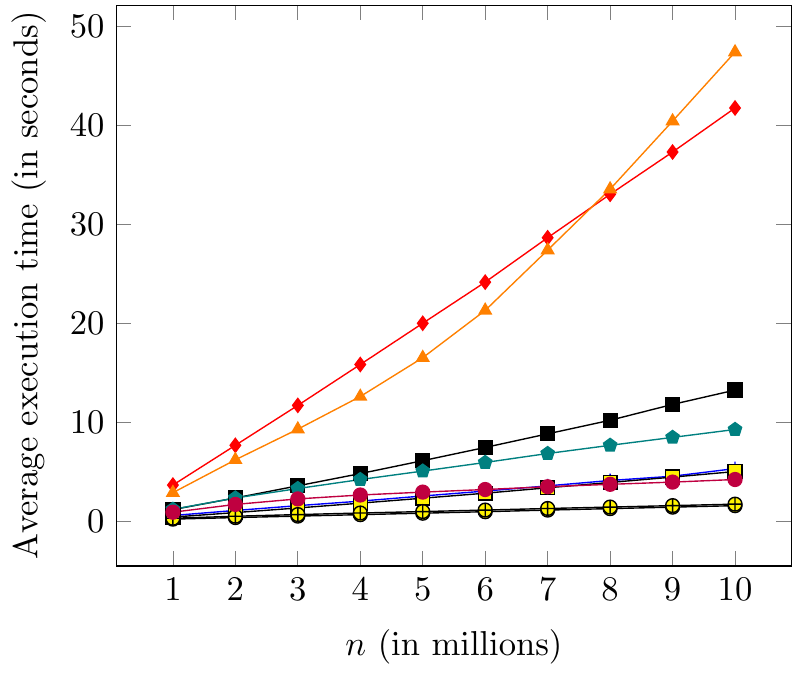} \hspace{55pt}	
	\includegraphics[scale=0.75]{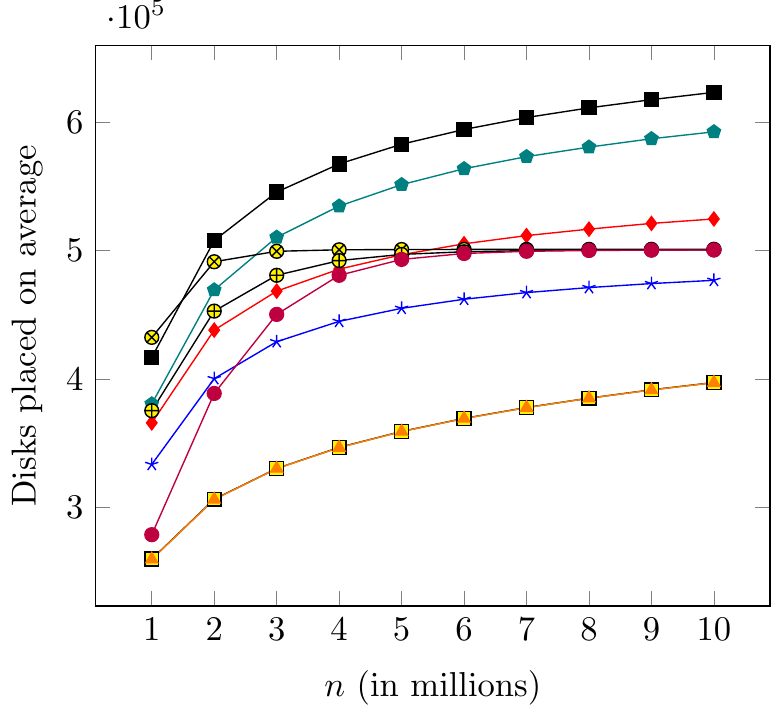}\\ \vspace*{20pt}
	\includegraphics[scale=0.58]{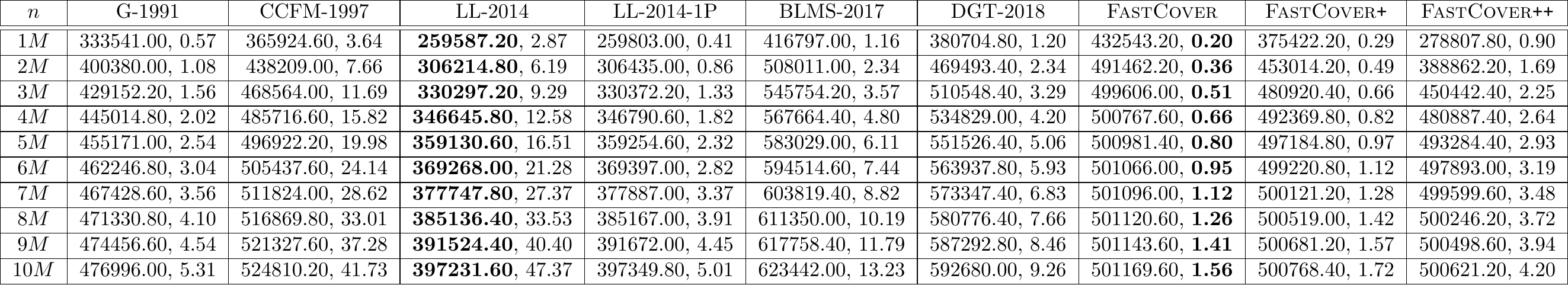}
	\caption{Points drawn from a disk of area $10^6$. A pair $x,y$ in a cell denotes the average number of disks placed and the average running time  in seconds, respectively.}
	\label{fig:PID-10-6}
\end{figure}

\begin{figure}[H]
	\centering
	\includegraphics[scale=0.75]{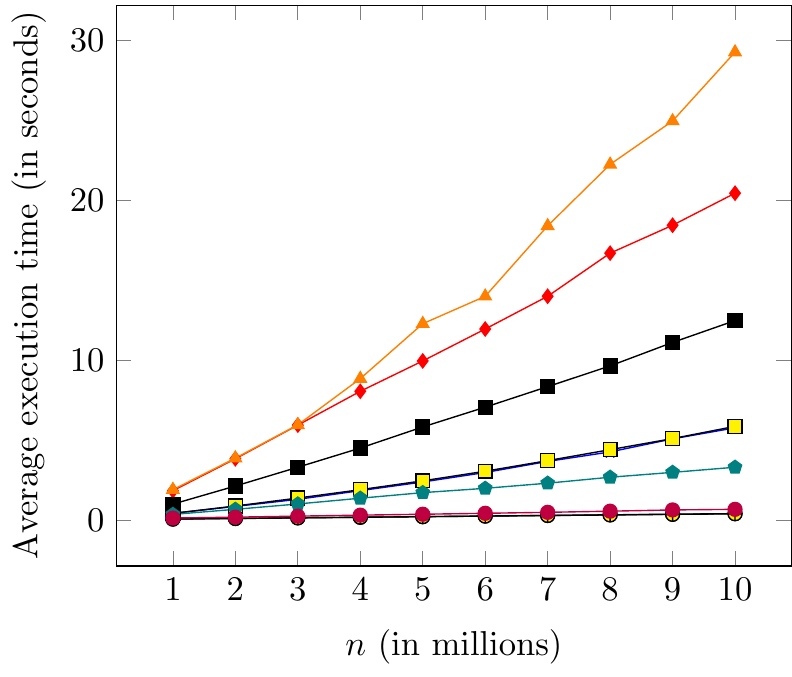} \hspace{55pt}	
	\includegraphics[scale=0.75]{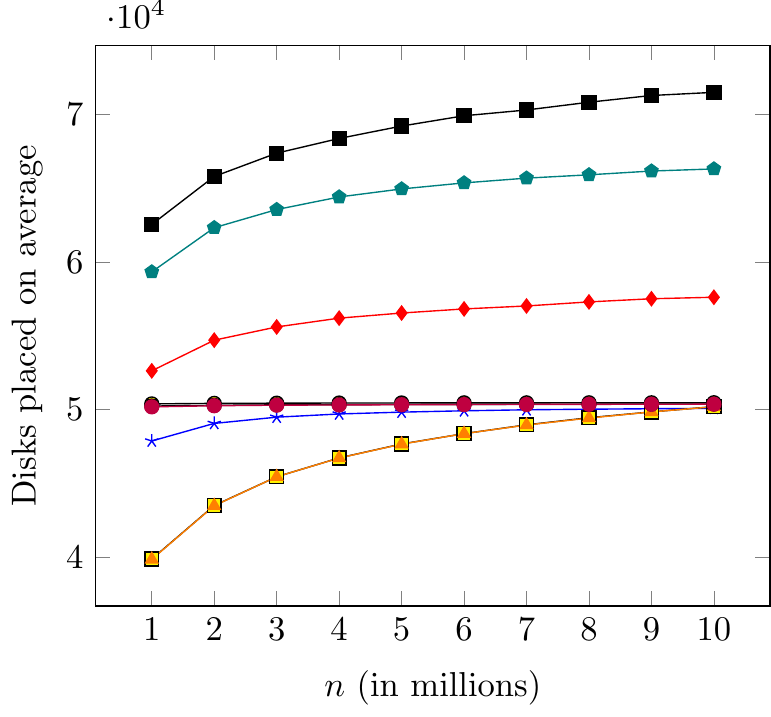}\\ \vspace*{20pt}
	\includegraphics[scale=0.58]{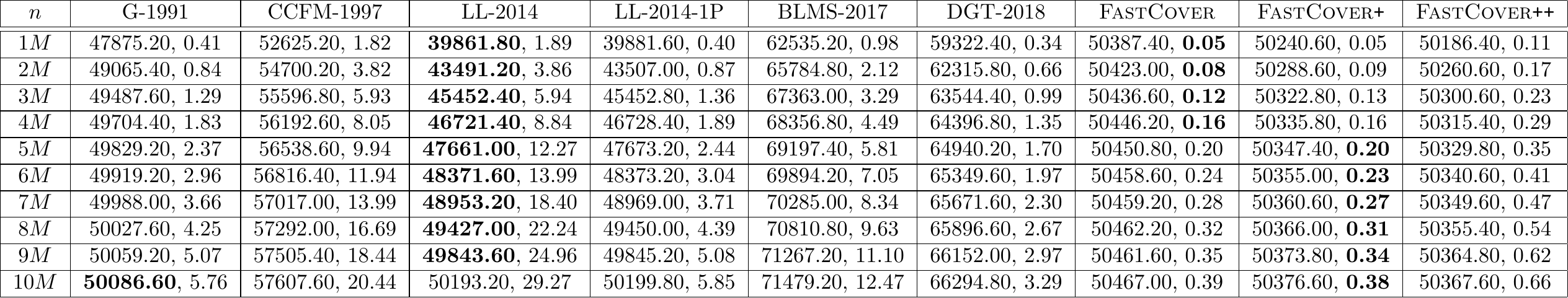}
	\caption{Points drawn from a disk of area $10^5$. A pair $x,y$ in a cell denotes the average number of disks placed and the average running time in seconds, respectively. }
	\label{fig:PID-10-5}
\end{figure}

\begin{figure}[H]
	\centering
	\includegraphics[scale=0.75]{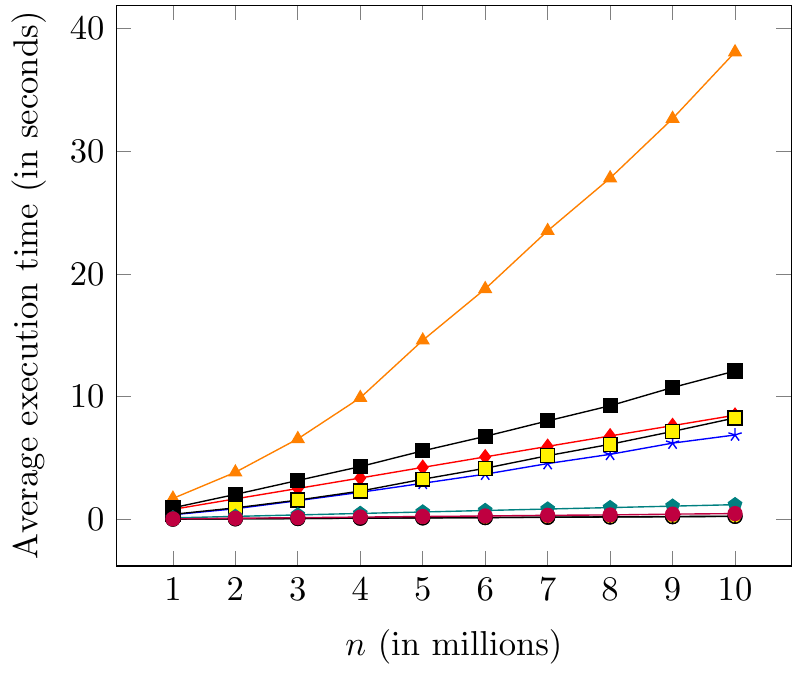} \hspace{55pt}	
	\includegraphics[scale=0.75]{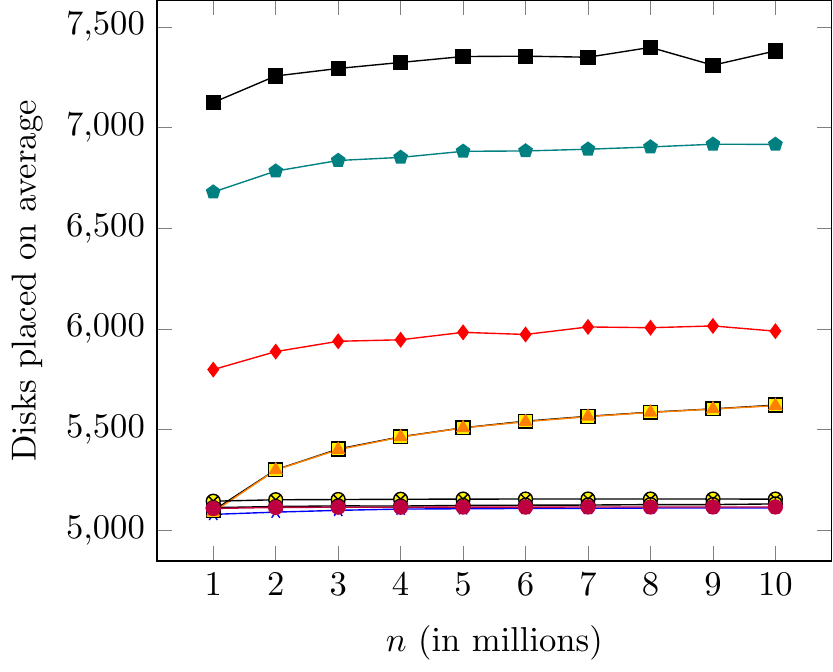}\\ \vspace*{20pt}
	\includegraphics[scale=0.58]{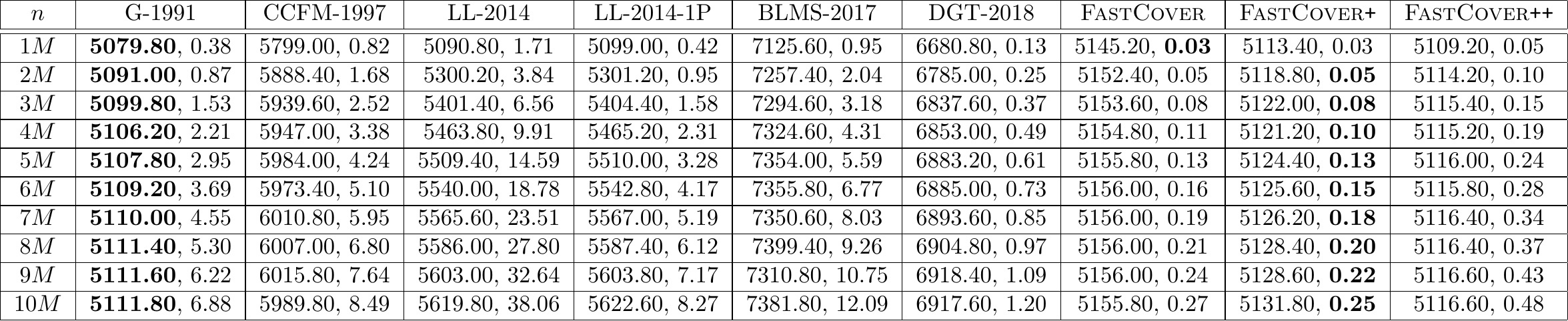}
	\caption{Points drawn from a disk of area $10^4$. A pair $x,y$ in a cell denotes the average number of disks placed and the average running time in seconds, respectively. }
	\label{fig:PID-10-4}
\end{figure}

\begin{figure}[h]
\centering
\includegraphics[scale=0.75]{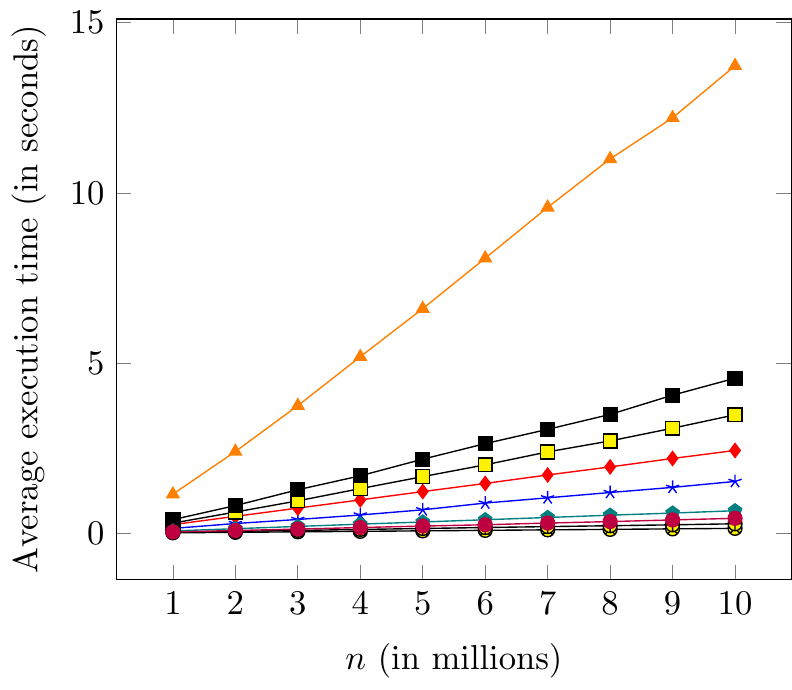} \hspace{55pt}	
\includegraphics[scale=0.75]{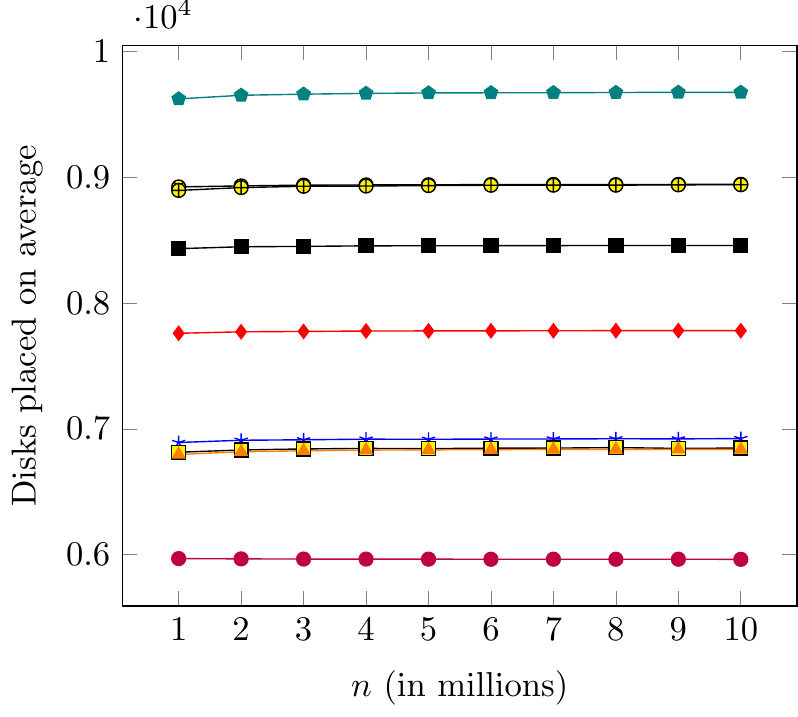}\\ \vspace*{20pt}
\includegraphics[scale=0.58]{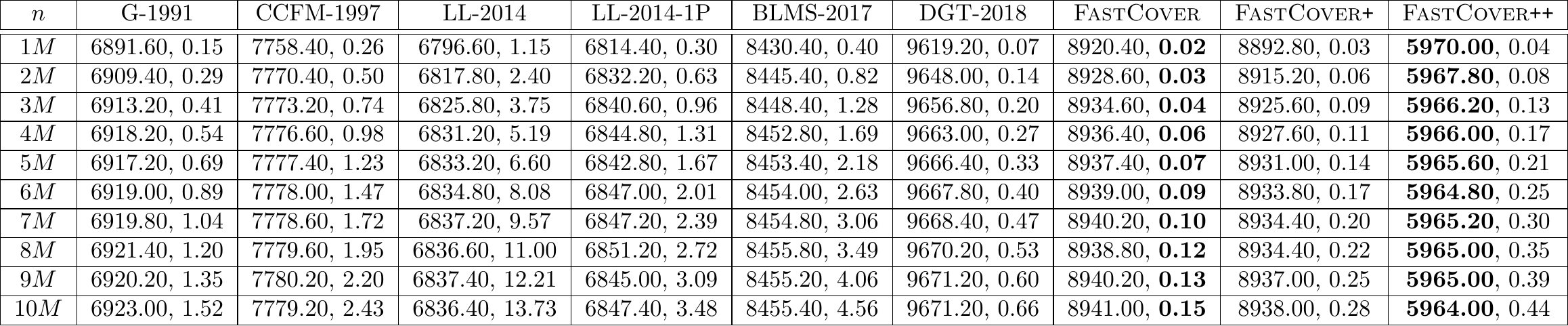}
\caption{Convex pointsets drawn from a square of area $10^7$. A pair $x,y$ in a cell denotes the average number of disks placed and the average running time in seconds, respectively.}
\label{fig:Conv-10-7}
\end{figure}

\begin{figure}[h]
	\centering
	\includegraphics[scale=0.75]{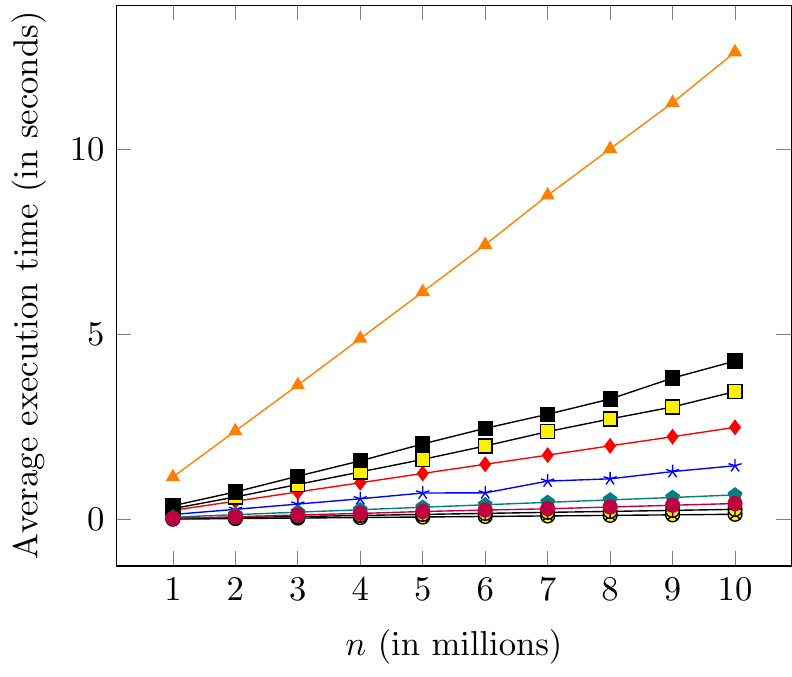} \hspace{55pt}	
	\includegraphics[scale=0.75]{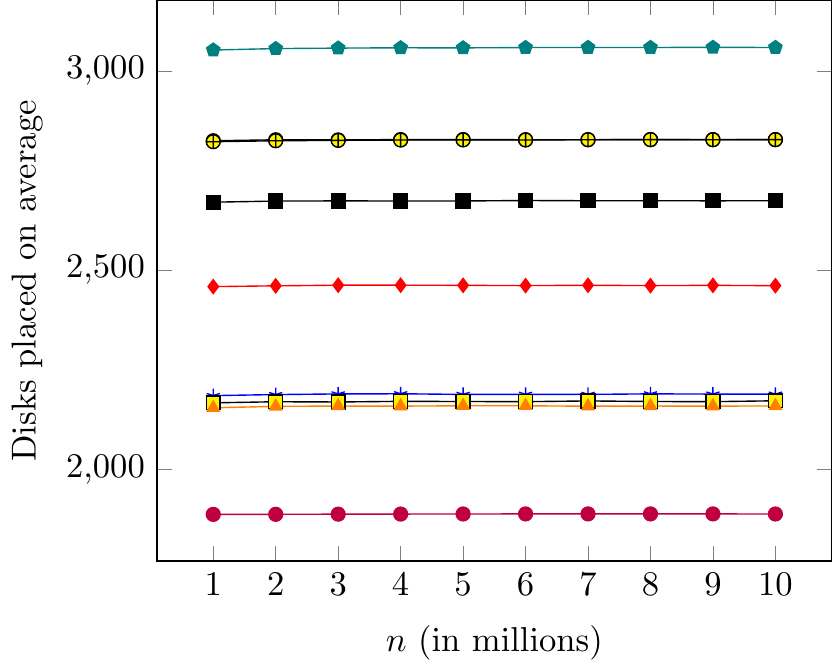}\\ \vspace*{20pt}
	\includegraphics[scale=0.58]{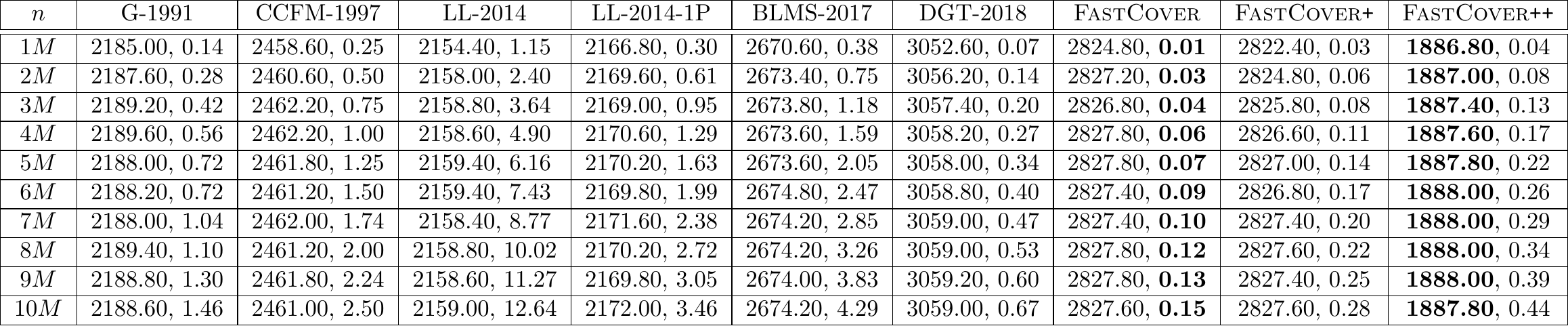}
	\caption{Convex pointsets drawn from a square of area $10^6$. A pair $x,y$ in a cell denotes the average number of disks placed and the average running time in seconds, respectively.}
	\label{fig:Conv-10-6}
\end{figure}

\begin{figure}[h]
	\centering
	\includegraphics[scale=0.75]{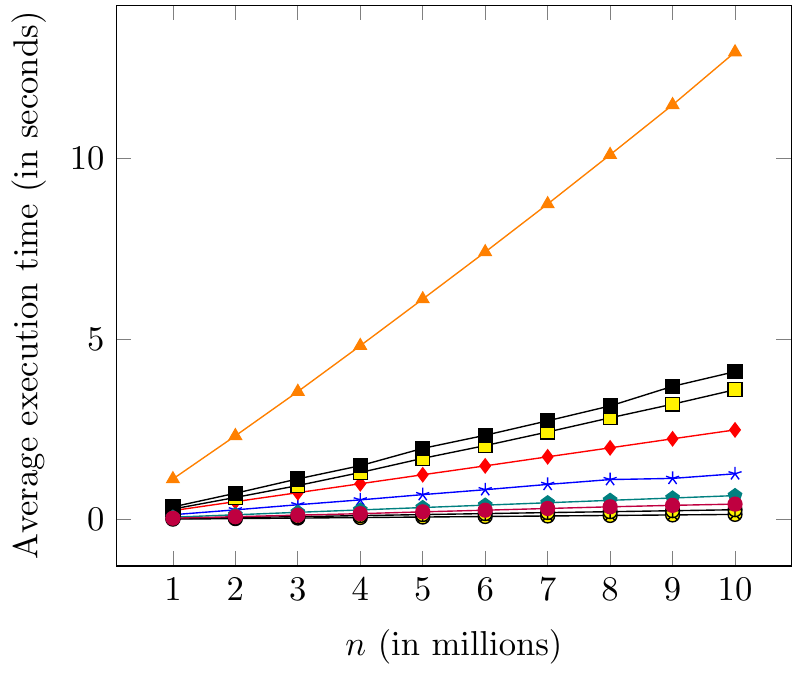} \hspace{55pt}	
	\includegraphics[scale=0.75]{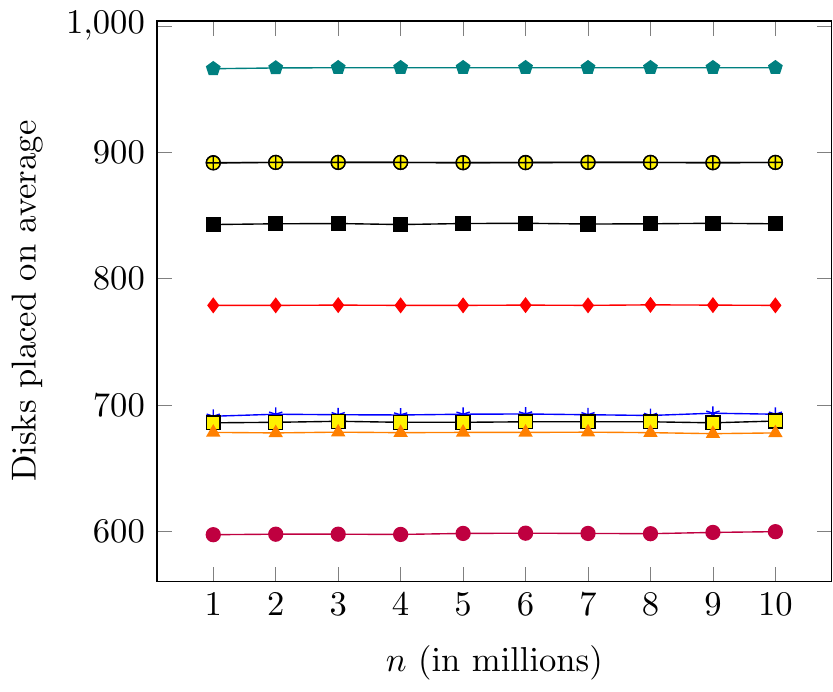}\\ \vspace*{20pt}
	\includegraphics[scale=0.58]{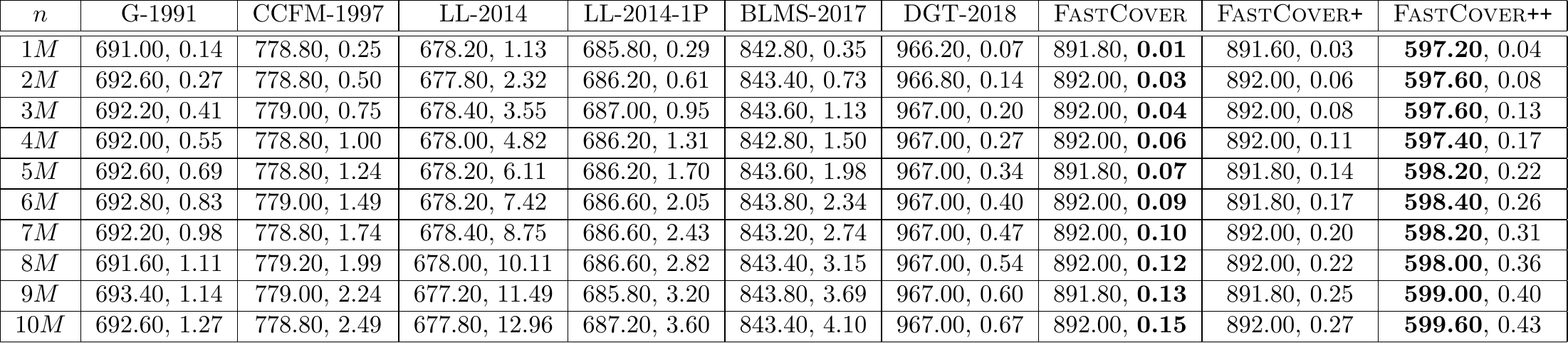}
	\caption{Convex pointsets drawn from a square of area $10^5$. A pair $x,y$ in a cell denotes the average number of disks placed and the average running time in seconds, respectively.}
	\label{fig:Conv-10-5}
\end{figure}

\begin{figure}[h]
	\centering
	\includegraphics[scale=0.75]{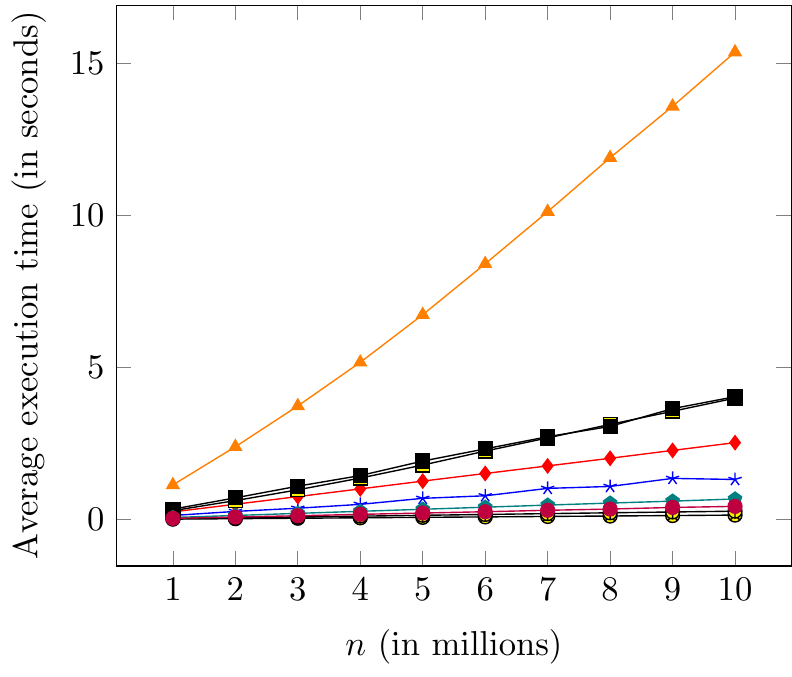} \hspace{55pt}	
	\includegraphics[scale=0.75]{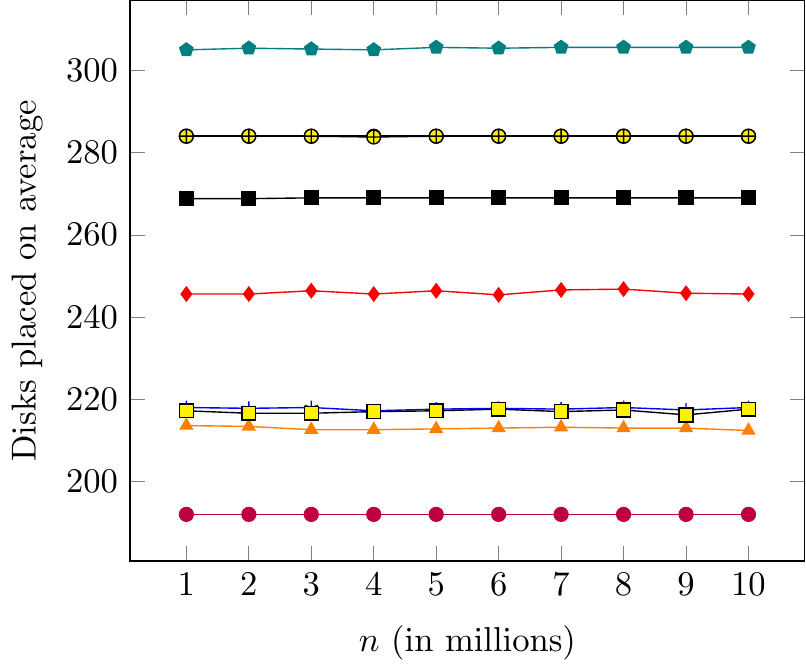}\\ \vspace*{20pt}
	\includegraphics[scale=0.58]{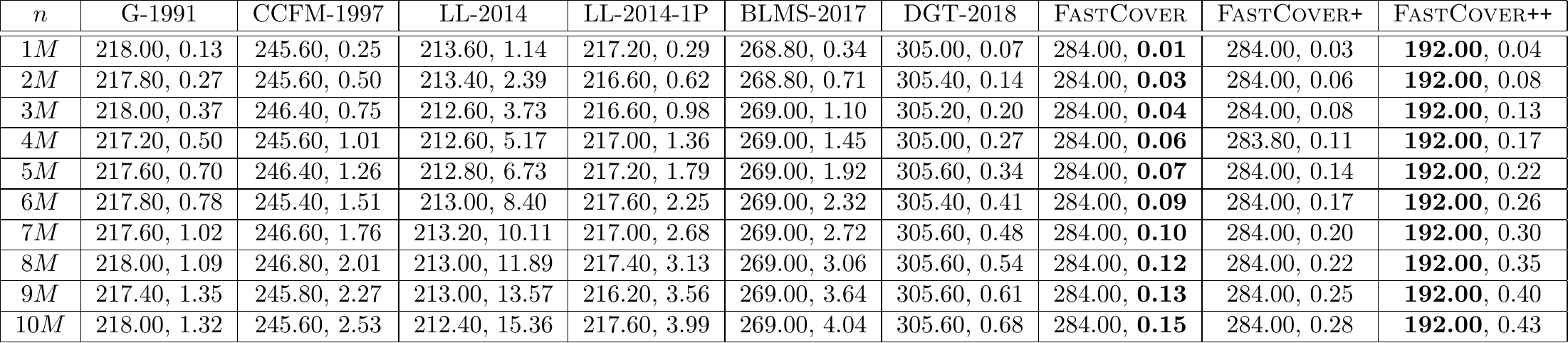}
	\caption{Convex pointsets drawn from a square of area $10^4$. A pair $x,y$ in a cell denotes the average number of disks placed and the average running time in seconds, respectively.}
	\label{fig:Conv-10-4}
\end{figure}

\newpage

\begin{figure}[h]
	\centering
\includegraphics[scale=0.75]{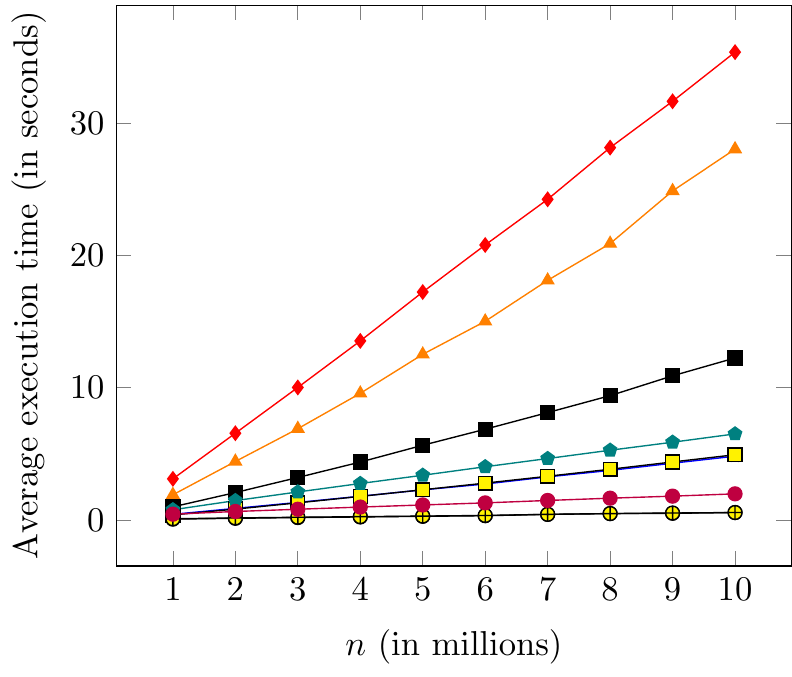} \hspace{55pt}	
\includegraphics[scale=0.75]{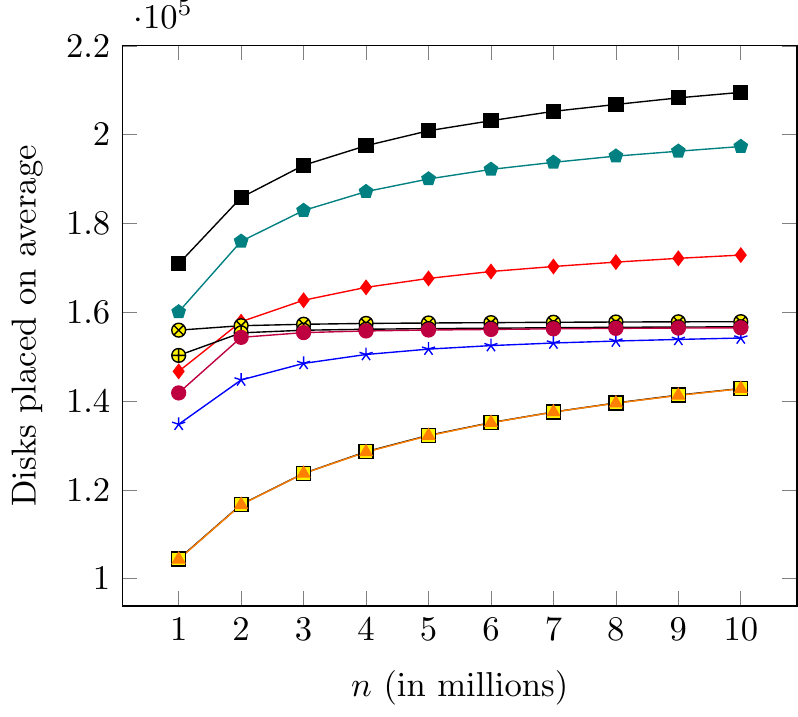}\\ \vspace*{20pt}
\includegraphics[scale=0.58]{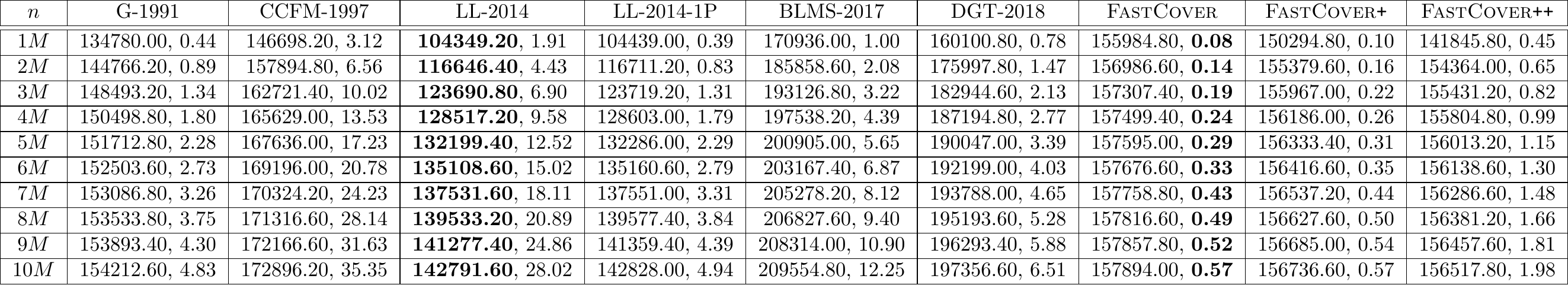}
\caption{Points drawn from an annulus. Radius of the outer circle is $10^3$ and that of the inner circle is $0.95\cdot 10^3$. A pair $x,y$ in a cell denotes the average number of disks placed and the average running time in seconds, respectively.}
	\label{fig:Ann-95}
\end{figure}

\begin{figure}[h]
	\centering
	\includegraphics[scale=0.75]{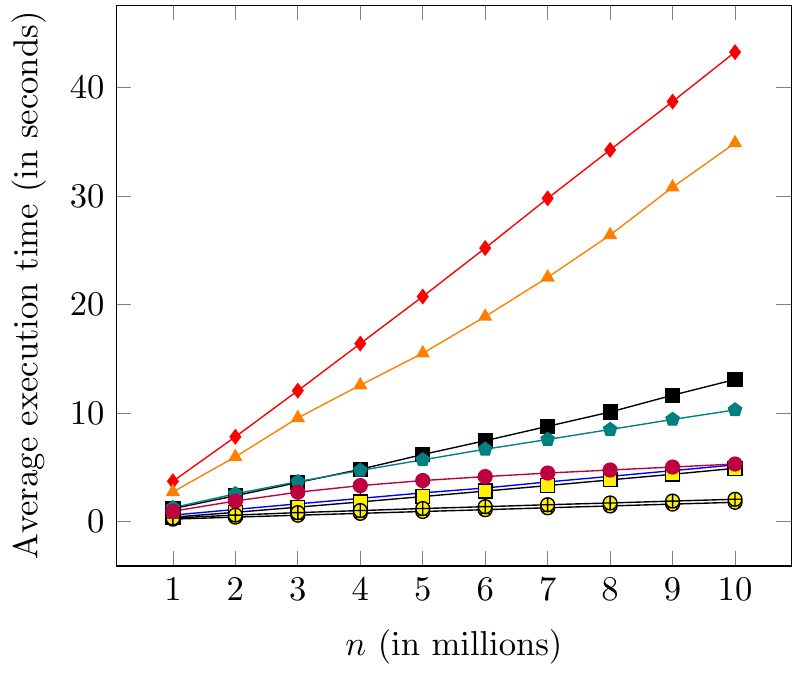} \hspace{55pt}	
	\includegraphics[scale=0.75]{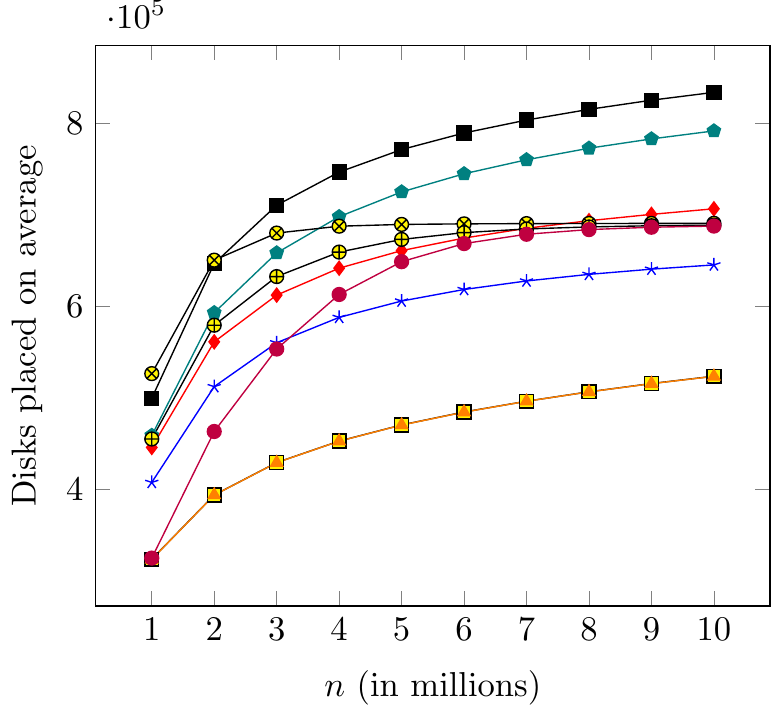}\\ \vspace*{20pt}
	\includegraphics[scale=0.58]{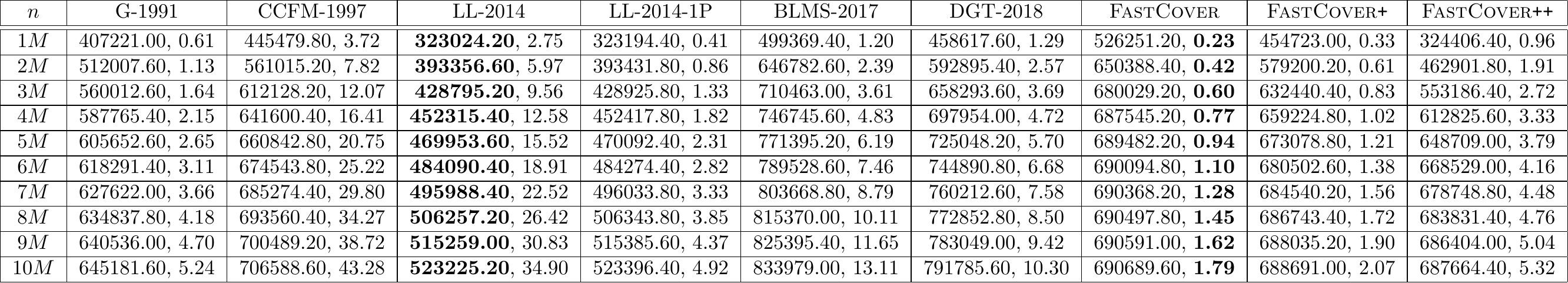}
	\caption{Points drawn from an annulus. Radius of the outer circle is $10^3$ and that of the inner circle is $0.75\cdot 10^3$. A pair $x,y$ in a cell denotes the average number of disks placed and the average running time in seconds, respectively.}
	\label{fig:Ann-75}
\end{figure}

\begin{figure}[h]
	\centering
	\includegraphics[scale=0.75]{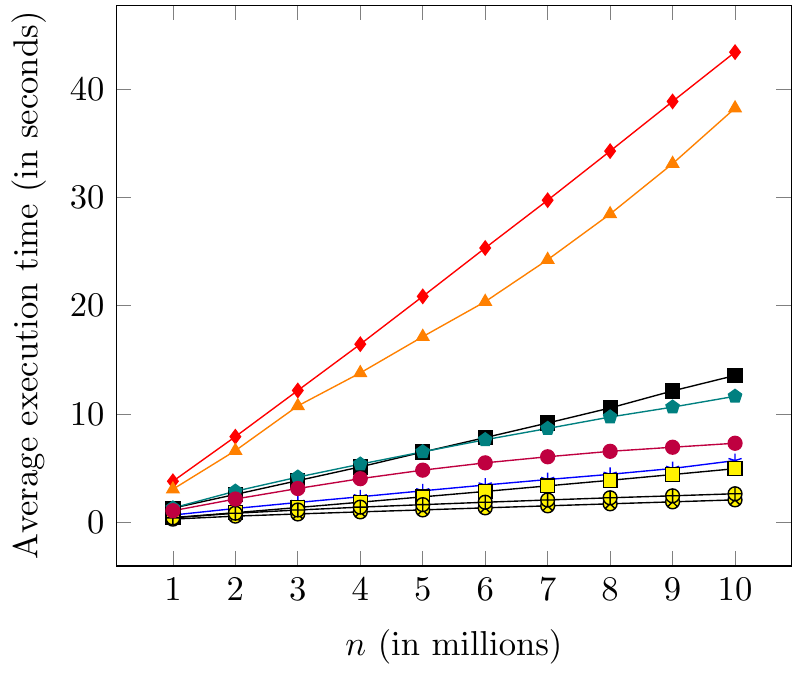} \hspace{55pt}	
	\includegraphics[scale=0.75]{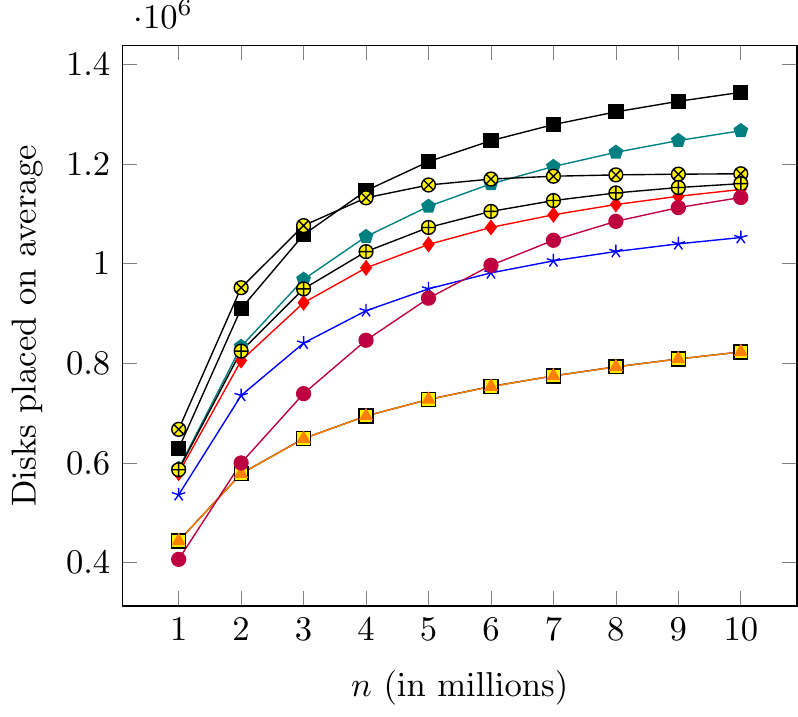}\\ \vspace*{20pt}
	\includegraphics[scale=0.58]{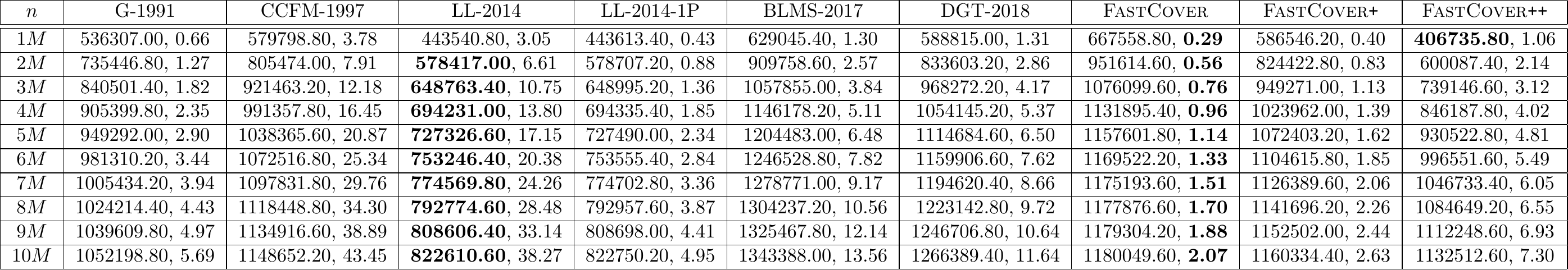}
	\caption{Points drawn from an annulus. Radius of the outer circle is $10^3$ and that of the inner circle is $0.5\cdot 10^3$. A pair $x,y$ in a cell denotes the average number of disks placed and the average running time in seconds, respectively.}
	\label{fig:Ann-50}
\end{figure}

\newpage
\clearpage

\section{A visual comparison of the engineered algorithms} \label{sec:demo}

\begin{figure}[h]
	\centering
	\includegraphics[scale=0.33]{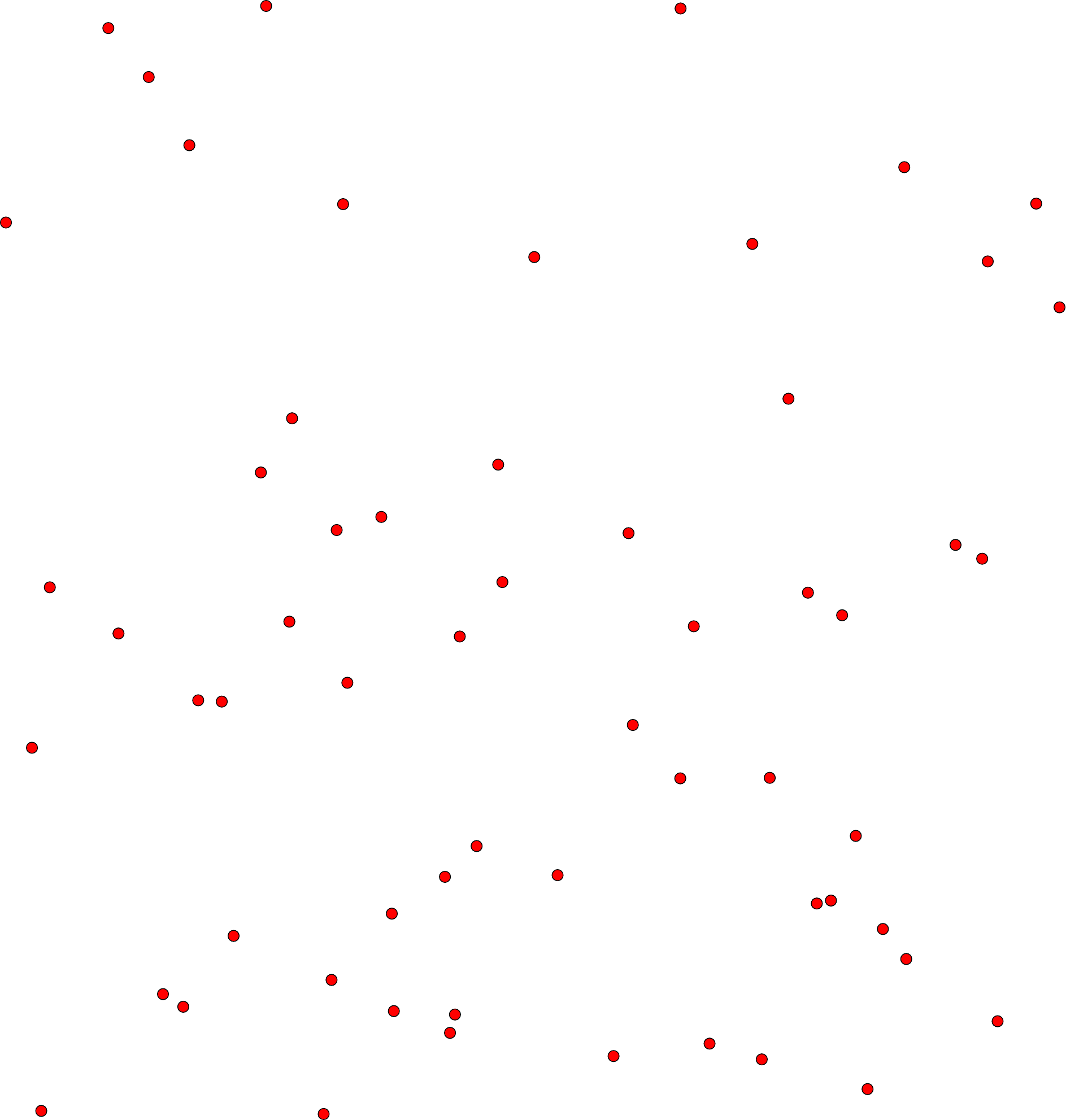}
	\caption{A $60$-element pointset drawn from a $20 \times 20$ square using \texttt{CGAL::Random\_points\_in\_square\_2}.}
	\label{fig:demopointset}
\end{figure}

\begin{figure}[h]
	\centering
	\includegraphics[scale=0.33]{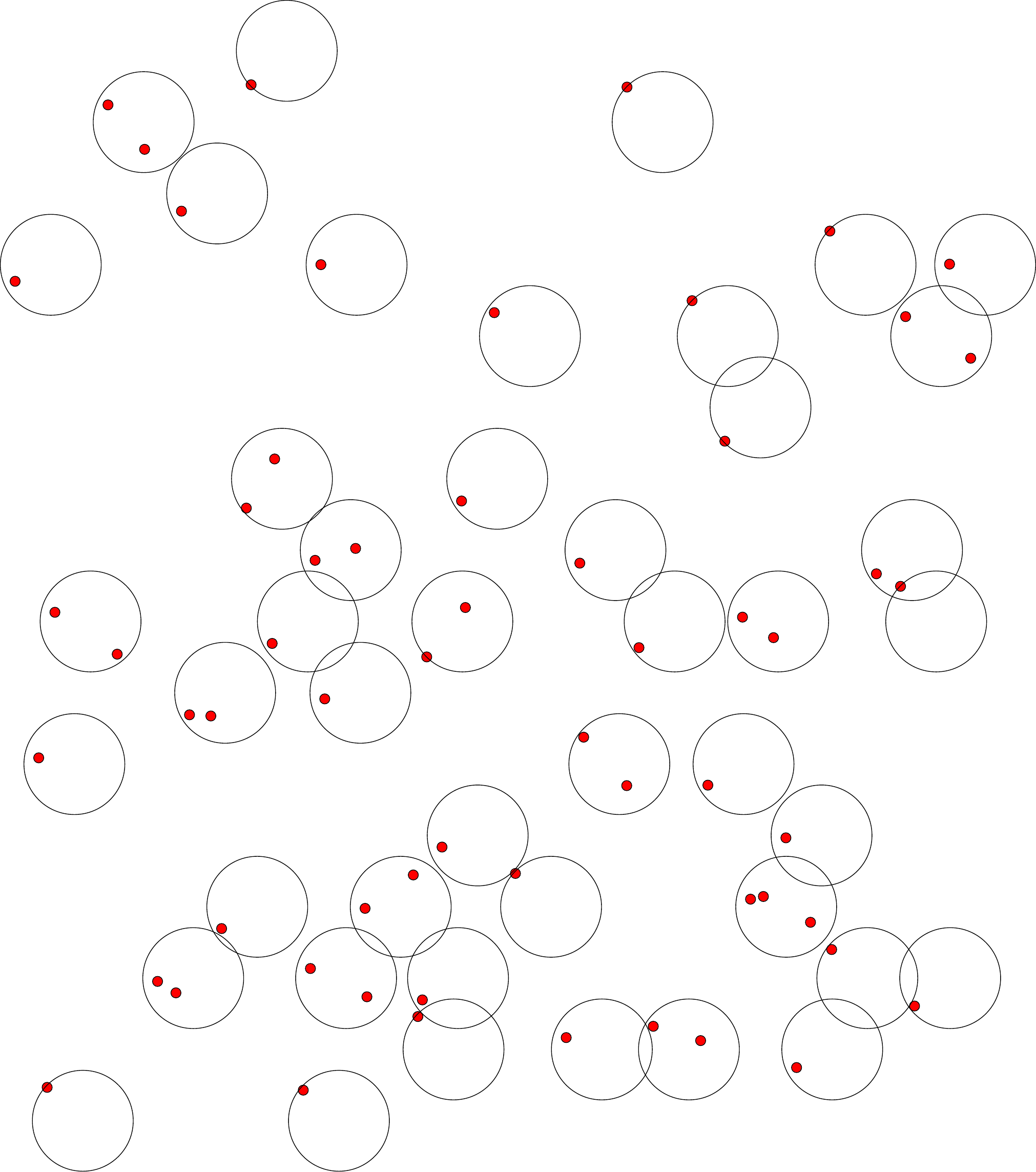}\hspace{50pt}	\includegraphics[scale=0.33]{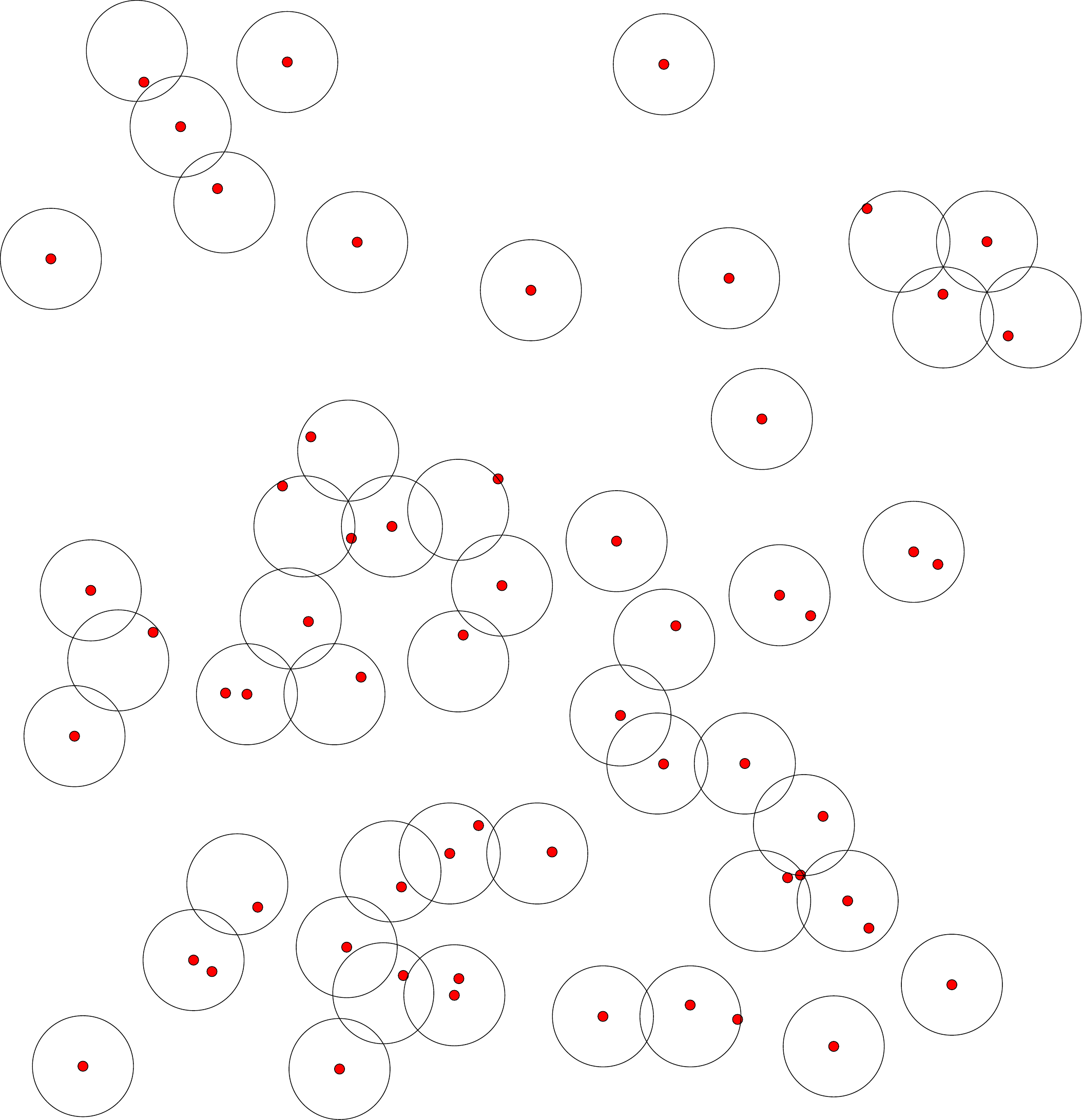}
	\caption{The outputs produced by G-1991(left) and CCFM-1997(right) on the pointset shown in Fig.~\ref{fig:demopointset}.}
\end{figure}

\begin{figure}[h]
	\centering
	\includegraphics[scale=0.33]{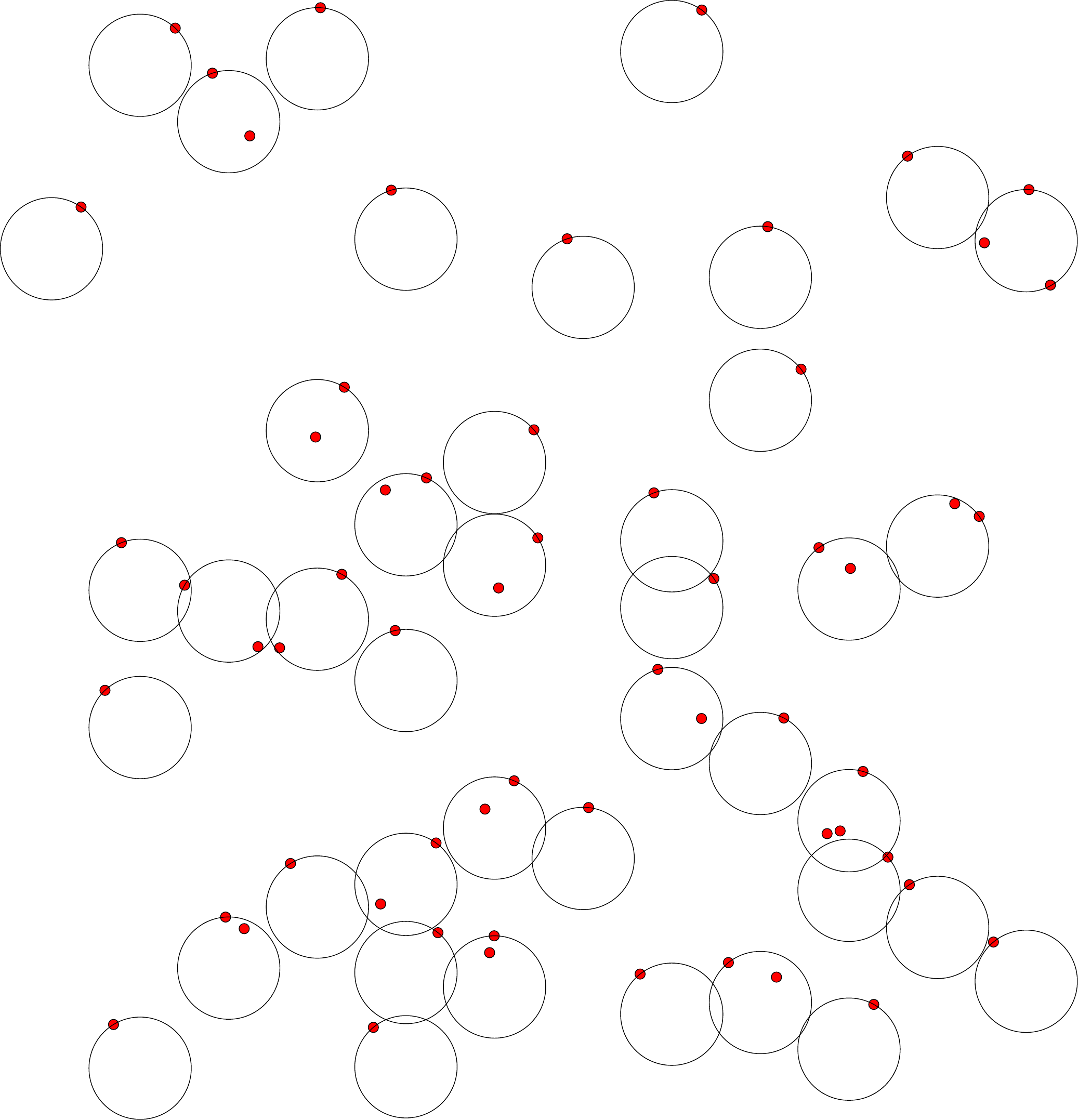}\hspace{50pt}	\includegraphics[scale=0.33]{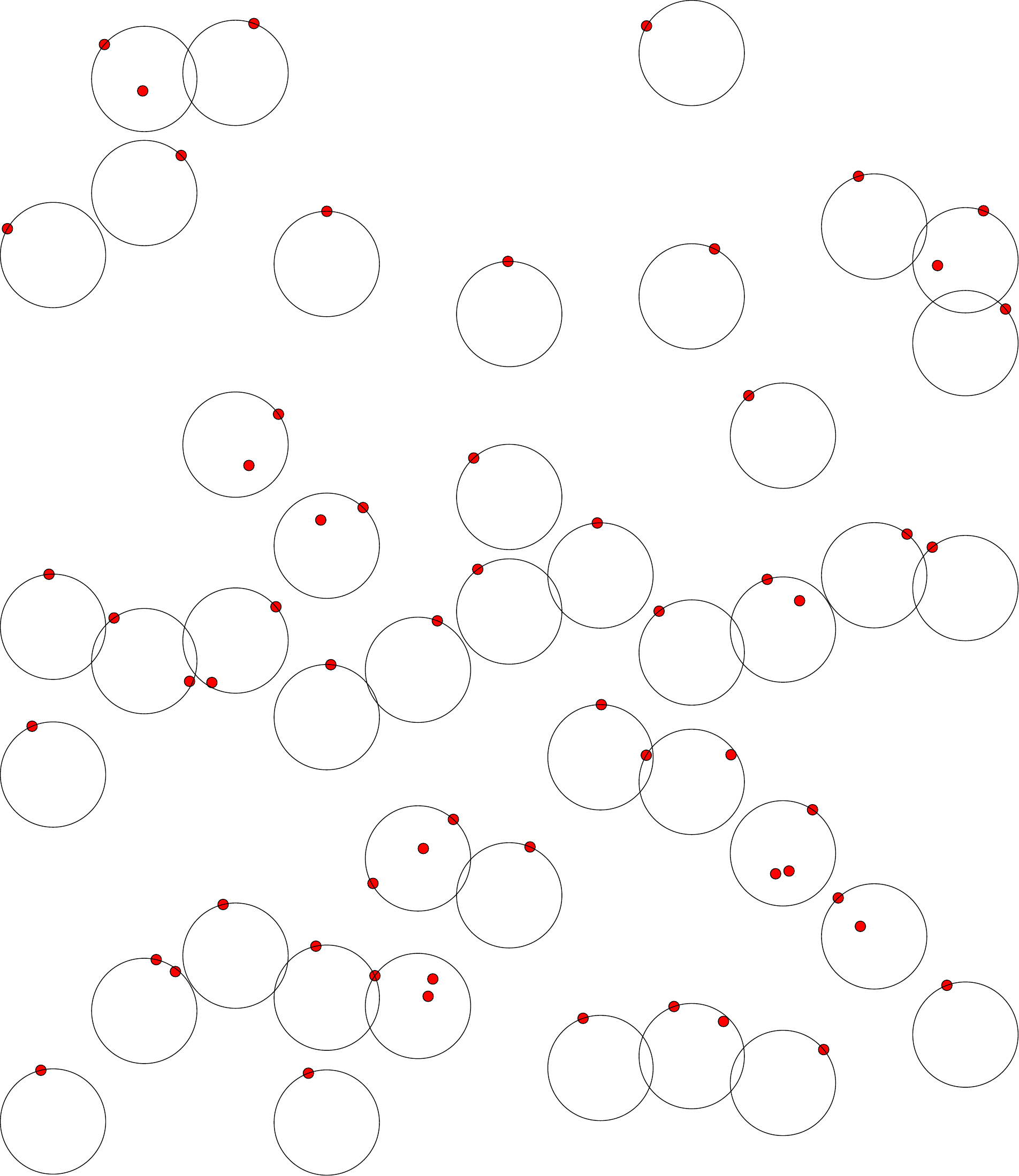}
	\caption{The outputs produced by LL-2014(left) and LL-2014-1P(right) on the pointset shown in Fig.~\ref{fig:demopointset}.}
\end{figure}

\begin{figure}[h]
	\centering
	\includegraphics[scale=0.33]{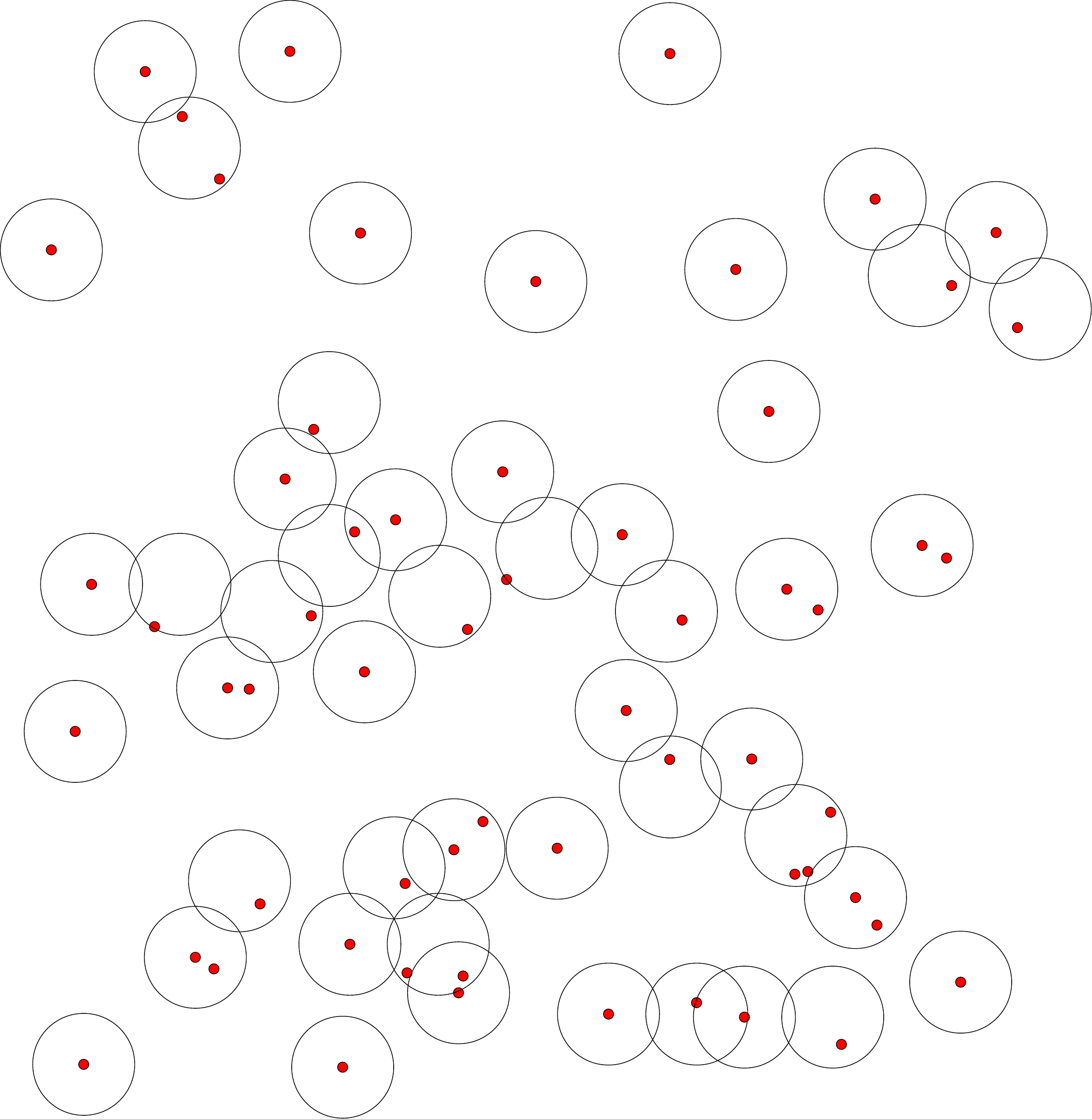}\hspace{50pt}	\includegraphics[scale=0.33]{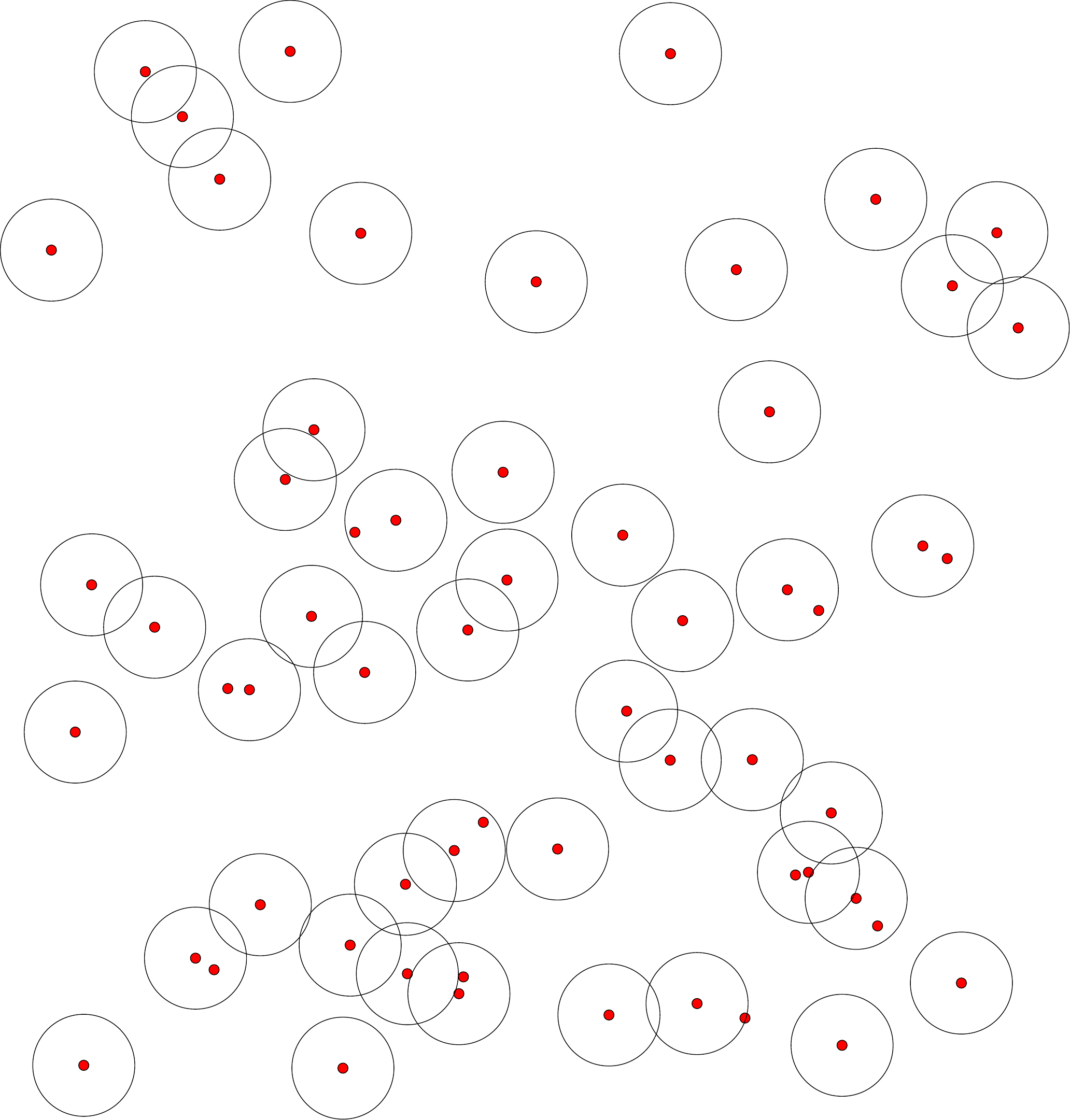}
	\caption{The  outputs produced by BLMS-2017(left) and DGT-2018(right) on the pointset shown in Fig.~\ref{fig:demopointset}.}
\end{figure}

\begin{figure}[h]
	\centering
	\includegraphics[scale=0.33]{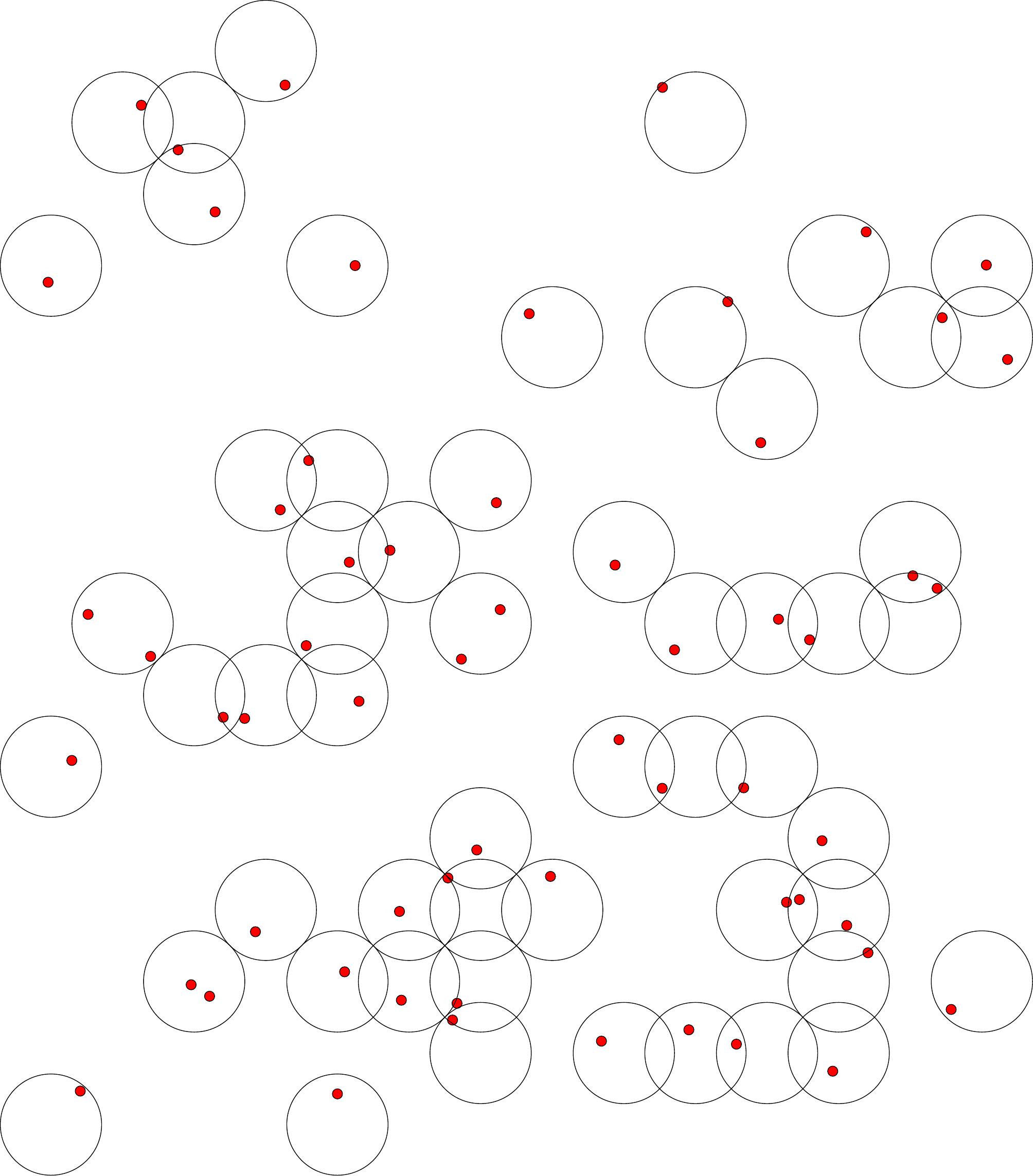}\hspace{50pt}	\includegraphics[scale=0.33]{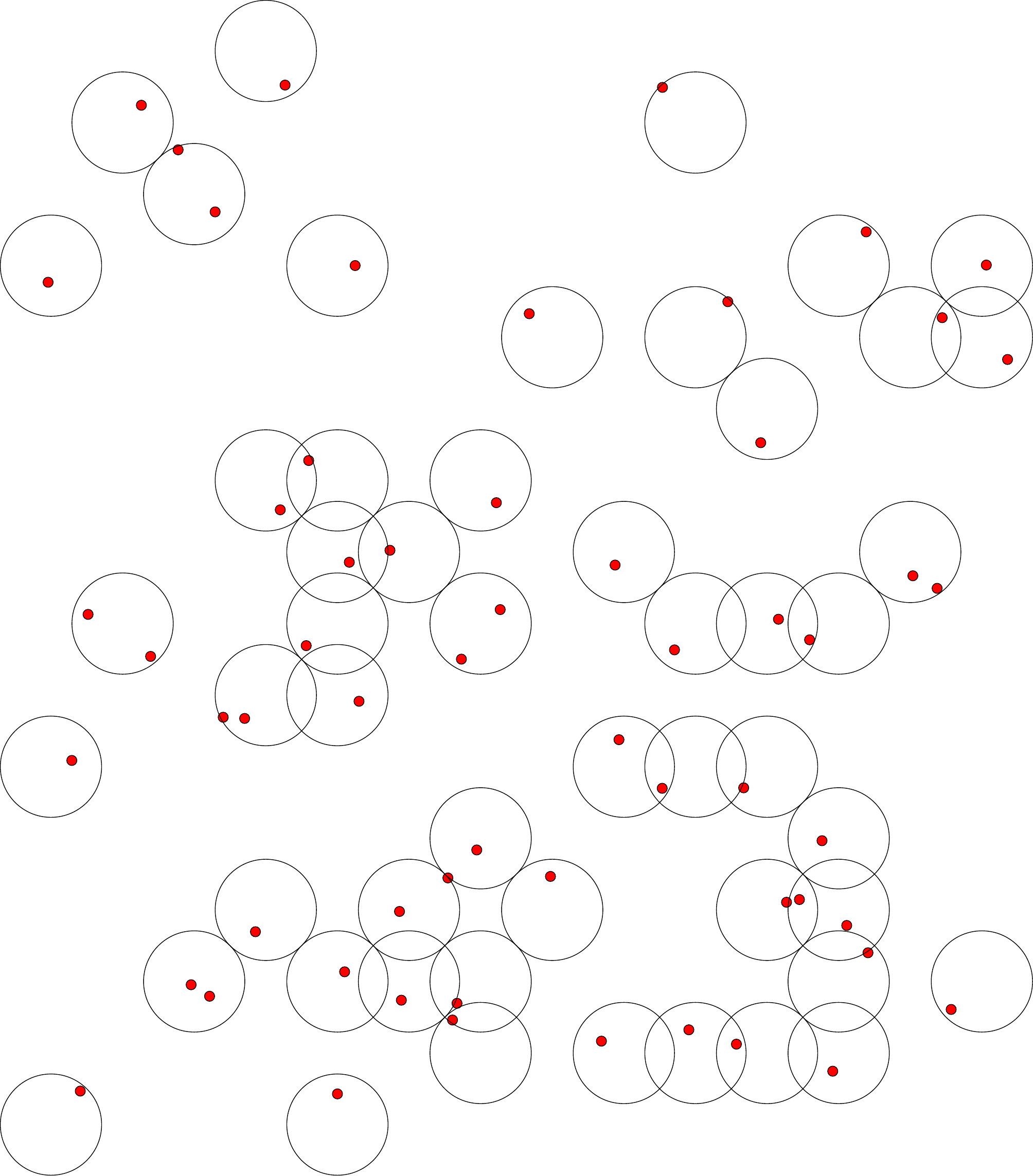}
\caption{The  outputs produced by  \textsc{FastCover}(left) and  \textsc{FastCover\texttt{+}}(right) on the pointset shown in Fig.~\ref{fig:demopointset}.}
\end{figure}

\begin{figure}[h]
	\centering
	\includegraphics[scale=0.33]{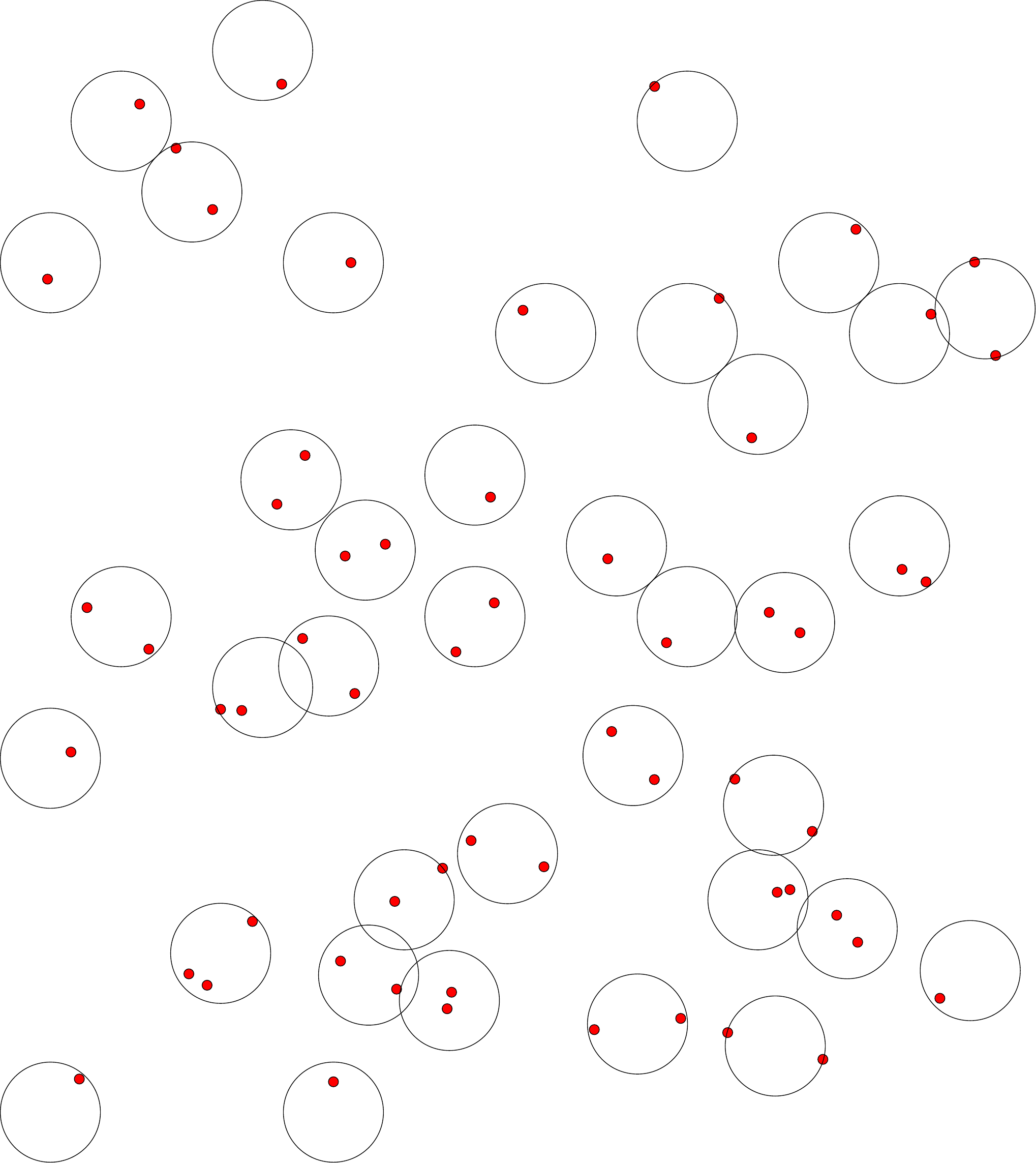}
	\caption{The output produced by \textsc{FastCover\texttt{++}} on the pointset shown in Fig.~\ref{fig:demopointset}.}
\end{figure}

\newpage
\clearpage

\bibliographystyle{spmpsci.bst}
\bibliography{UDCbib}

\begin{thebibliography}{10}
\providecommand{\url}[1]{{#1}}
\providecommand{\urlprefix}{URL }
\expandafter\ifx\csname urlstyle\endcsname\relax
  \providecommand{\doi}[1]{DOI~\discretionary{}{}{}#1}\else
  \providecommand{\doi}{DOI~\discretionary{}{}{}\begingroup
  \urlstyle{rm}\Url}\fi

\bibitem{tsp}
\url{www.math.uwaterloo.ca/tsp/}

\bibitem{agarwal2014near}
Agarwal, P.K., Pan, J.: Near-linear algorithms for geometric hitting sets and
  set covers.
\newblock In: Proceedings of the thirtieth annual symposium on Computational
  geometry, p. 271. ACM (2014)

\bibitem{aloupis2012covering}
Aloupis, G., Hearn, R.A., Iwasawa, H., Uehara, R.: Covering points with
  disjoint unit disks.
\newblock In: CCCG, pp. 41--46 (2012)

\bibitem{bar2013note}
Bar-Yehuda, R., Rawitz, D.: A note on multicovering with disks.
\newblock Computational Geometry \textbf{46}(3), 394--399 (2013)

\bibitem{biniaz2017approximation}
Biniaz, A., Liu, P., Maheshwari, A., Smid, M.: Approximation algorithms for the
  unit disk cover problem in 2\textsc{D} and 3\textsc{D}.
\newblock Computational Geometry \textbf{60}, 8--18 (2017)

\bibitem{bronnimann1995almost}
Br{\"o}nnimann, H., Goodrich, M.T.: Almost optimal set covers in finite
  \textsc{VC}-dimension.
\newblock Discrete \& Computational Geometry \textbf{14}(4), 463--479 (1995)

\bibitem{bus2018practical}
Bus, N., Mustafa, N.H., Ray, S.: Practical and efficient algorithms for the
  geometric hitting set problem.
\newblock Discrete Applied Mathematics \textbf{240}, 25--32 (2018)

\bibitem{charikar2004incremental}
Charikar, M., Chekuri, C., Feder, T., Motwani, R.: Incremental clustering and
  dynamic information retrieval.
\newblock SIAM Journal on Computing \textbf{33}(6), 1417--1440 (2004)

\bibitem{chazelle1986circle}
Chazelle, B.M., Lee, D.T.: On a circle placement problem.
\newblock Computing \textbf{36}(1-2), 1--16 (1986)

\bibitem{das2012discrete}
Das, G.K., Fraser, R., L{\'o}opez-Ortiz, A., Nickerson, B.G.: On the discrete
  unit disk cover problem.
\newblock International Journal of Computational Geometry \& Applications
  \textbf{22}(05), 407--419 (2012)

\bibitem{de2009covering}
De~Berg, M., Cabello, S., Har-Peled, S.: Covering many or few points with unit
  disks.
\newblock Theory of Computing Systems \textbf{45}(3), 446--469 (2009)

\bibitem{dumitrescu2018computational}
Dumitrescu, A.: Computational geometry column 68.
\newblock ACM SIGACT News \textbf{49}(4), 46--54 (2018)

\bibitem{dumitrescu2020online}
Dumitrescu, A., Ghosh, A., T{\'o}th, C.D.: Online unit covering in euclidean
  space.
\newblock Theoretical Computer Science \textbf{809}, 218--230 (2020)

\bibitem{fowler1981optimal}
Fowler, R.J.: Optimal packing and covering in the plane are
  \textsc{NP}-complete.
\newblock Inf. Process. Lett \textbf{12}(3), 133--137 (1981)

\bibitem{franceschetti2001geometric}
Franceschetti, M., Cook, M., Bruck, J.: A geometric theorem for approximate
  disk covering algorithms  (2001)

\bibitem{fu2007almost}
Fu, B., Chen, Z., Abdelguerfi, M.: An almost linear time 2.8334-approximation
  algorithm for the disc covering problem.
\newblock In: International Conference on Algorithmic Applications in
  Management, pp. 317--326. Springer (2007)

\bibitem{ghosh2019unit}
Ghosh, A., Hicks, B., Shevchenko, R.: Unit disk cover for massive point sets.
\newblock In: International Symposium on Experimental Algorithms, pp. 142--157.
  Springer (2019)

\bibitem{gonzalez1991covering}
Gonzalez, T.F.: Covering a set of points in multidimensional space.
\newblock Information processing letters \textbf{40}(4), 181--188 (1991)

\bibitem{guo2021geometric}
Guo, Z., Li, Y.: Geometric cover with outliers removal.
\newblock In: 38th International Symposium on Theoretical Aspects of Computer
  Science (STACS 2021). Schloss Dagstuhl-Leibniz-Zentrum f{\"u}r Informatik
  (2021)

\bibitem{hochbaum1985approximation}
Hochbaum, D.S., Maass, W.: Approximation schemes for covering and packing
  problems in image processing and \textsc{VLSI}.
\newblock Journal of the ACM (JACM) \textbf{32}(1), 130--136 (1985)

\bibitem{kaplan2011optimal}
Kaplan, H., Katz, M.J., Morgenstern, G., Sharir, M.: Optimal cover of points by
  disks in a simple polygon.
\newblock SIAM Journal on Computing \textbf{40}(6), 1647--1661 (2011)

\bibitem{liao2010polynomial}
Liao, C., Hu, S.: Polynomial time approximation schemes for minimum disk cover
  problems.
\newblock Journal of combinatorial optimization \textbf{20}(4), 399--412 (2010)

\bibitem{liaw2017}
Liaw, C., Liu, P., Reiss, R.: Approximation schemes for covering and packing in
  the streaming model.
\newblock In: Canadian Conference on Computational Geometry (2018)

\bibitem{liu2014fast}
Liu, P., Lu, D.: A fast 25/6-approximation for the minimum unit disk cover
  problem.
\newblock arXiv preprint arXiv:1406.3838  (2014)

\bibitem{cgal:eb-21b}
{The CGAL Project}: {CGAL} User and Reference Manual, {5.3} edn.
\newblock {CGAL Editorial Board} (2021).
\newblock \urlprefix\url{https://doc.cgal.org/5.3/Manual/packages.html}

\end{thebibliography}

\end{document}